\documentclass[english]{amsart}
\usepackage[T1]{fontenc}
\usepackage[latin1]{inputenc}
\usepackage{geometry}
\usepackage{float}
\usepackage{mathrsfs}
\usepackage{amsmath}
\usepackage{amssymb}
\usepackage{graphicx}
\usepackage{subfigure}

\makeatletter

\floatstyle{ruled}
\newfloat{algorithm}{tbp}{loa}
\providecommand{\algorithmname}{Algorithm}
\floatname{algorithm}{\protect\algorithmname}

\usepackage{amsthm}

\usepackage{amsfonts}

\usepackage{epsfig}

\usepackage{bm}

\usepackage{mathrsfs}

\usepackage{enumerate}

\@ifundefined{definecolor}{\@ifundefined{definecolor}
 {\@ifundefined{definecolor}
 {\usepackage{color}}{}
}{}
}{}

\usepackage{subfig}\usepackage[all]{xy}

\newcommand{\bbR}{\mathbb R}
\newcommand{\bbC}{\mathbb C}
\newcommand{\bbE}{\mathbb E}
\newcommand{\bbP}{\mathbb P}

\newcommand{\cO}{\mathcal O}

\newtheorem{theorem}{Theorem}[section]

\newtheorem{rem}{Remark}[section]
\newtheorem{prop}{Proposition}[section]

\newtheorem{ass}{Assumption}[section]
\newcounter{hypA}

\newcounter{probP}
\newenvironment{probP}{\refstepcounter{probP}\begin{itemize}
  \item[({\bf P\arabic{probP}})]}{\end{itemize}}

\usepackage{babel}\date{}

\author{Ryan Bennink$^*$}
\thanks{$^*$Oak Ridge National Laboratory, Oak Ridge, TN, 37831}

\author{Ajay Jasra$^+$}
\thanks{$^+$Department of Statistics and Applied Probability, National University of Singapore, Singapore}

\author{Kody J. H. Law$^\dagger$}
\thanks{$^\dagger$School of Mathematics, University of Manchester, Manchester, UK, M13 4PL}

\author{Pavel Lougovski$^\$$}
\thanks{$^\$$Oak Ridge National Laboratory, Oak Ridge, TN, 37831}

\makeatother

\title[UQ for quantum simulators]
{Estimation and uncertainty quantification for the output from quantum simulators}

\usepackage{babel}
\begin{document}

\begin{abstract}

The problem of estimating certain distributions over $\{0,1\}^d$ is considered here.
The distribution represents a quantum system of $d$ qubits, 
where there are non-trivial dependencies between the qubits.  
A maximum entropy approach is adopted to reconstruct the 
distribution from exact moments or observed empirical moments. 
The Robbins Monro algorithm is used to solve the intractable maximum entropy problem, 
by constructing an unbiased estimator of the un-normalized target with 
a sequential Monte Carlo sampler at each iteration.
In the case of empirical moments, this coincides with a maximum likelihood estimator.  
A Bayesian formulation is also considered in order to quantify posterior uncertainty.  
Several approaches are proposed in order to tackle this challenging problem, based on recently developed methodologies.
In particular, unbiased estimators of the gradient of the log posterior are constructed and used within a
provably convergent Langevin-based Markov chain Monte Carlo method.
The methods are illustrated on classically simulated 
output from quantum simulators.

\end{abstract}

\maketitle

\section{Introduction and motivation}\label{introduction}

Quantum computing holds a great promise to enable future simulations of quantum systems that would otherwise suffer from the curse of dimensionality when simulated classically. 
However, we are currently entering  the Noisy Intermediate Scale Quantum (NISQ) device era~\cite{preskillNISQ2018}. The hallmarks of NISQ era are relatively small (10-50 qubits) quantum computing devices and noisy qubits. Therefore, to succeed with quantum simulation on NISQ quantum computers, it is imperative to understand and fully characterize the effects of noise on the state of a quantum device. 

In quantum mechanics the state of a $d$-qubit quantum device is generally described by a density matrix $\rho$ --  a $2^{d}$-by-$2^{d}$ complex-valued positive semidefinite matrix with ${\rm Tr}\rho=1$. From the definition, it takes $2^{2d}-1$ real parameters to fully specify an arbitrary quantum state. If $O$ denotes a quantum observable, described by a $2^{d}$-by-$2^{d}$ Hermitian matrix ($O^{\dagger}=O$),
then measuring $O$ yields one of its eigenvalues sampled from a distribution determined by $\rho$. In particular, the probability of the $i$th eigenvalue is $p_{i}=\langle \psi_i|\rho|\psi_i\rangle$, where $|\psi_i\rangle$ is the $i$th eigenvector in the bra-ket notation. Since there are at most $2^d$ eigenvalues, repeated measurements of any observable yields information about at most $2^{d}-1$ parameters of $\rho$. To estimate every parameter of the state one must measure at least $2^{d}+1$ independent observables.
This approach is known as quantum state tomography~\cite{Paris:QSE2010,hradil:QSE1997}. Its applicability is limited to small quantum devices due to exponential scaling in the number of measurements and  related issues of solving exponentially large systems of equations. Bayesian methods also have been suggested~\cite{jonesQInference1991,blumeBME:2010,huszarAdaptiveBQT2012,granadeQBME2016} to facilitate uncertainty quantification (UQ) in the context of quantum state tomography. However, the existing Bayesian quantum tomography solutions face similar scalability challenge as the number of parameters to infer still grows exponentially with the size of a quantum device.  


Here we consider a fundamentally new approach, which 
provides a scalable solution to  
UQ for the output of quantum simulations. 
We will restrict our attention to a simplified scenario
in which the density matrix and observable commute, 
hence are simultaneously diagonalizable, and the eigenvectors are known.
In this way we are able to treat the quantum state, which is represented by 
the density matrix, as a classical probability distribution, 
given by the diagonal of the density matrix. 
Note that the standard quantum state tomography methods would still scale 
exponentially even for this simplified scenario.
To avoid this problem we use the {\it maximum entropy} (MaxEnt) principle~\cite{jaynesMaxEnt1:1957,jaynesMaxEnt2:1957} 
to model a complete probability distribution based on a selected subset of observed or exact moments.
We then use well-established stochastic simulation algorithms \cite{robbins1951stochastic, kushner2003stochastic, delm:04, del2006sequential} 
to find the MaxEnt distribution with a cost which is polynomial in $d$.
Adopting a MaxEnt ansatz for the likelihood provides a Bayesian posterior which 
(i) converges to the MaxEnt distribution as the number of observations tends to infinity, and 
(ii) can be simulated with a cost which is polynomial in $d$ using recent stochastic simulation methods \cite{mcleish2011general, rhee2015unbiased, teh2011bayesian, teh2016consistency}.
%
It is the topic of ongoing investigation to generalize the results to density matrices,
which will be reported in a future publication. 
It is hoped that a polynomial cost in $d$ can be preserved in this case as well, 
for example combining the methodology here with that of \cite{gross2010quantum,candes2009exact}.


Assume that we know the basis in which the output state from a $d$-qubit quantum computer is diagonal. The density matrix in that basis reads,
\begin{equation}\label{eq:density}
\rho =  \sum_{j=1}^{2^d} p_j |\psi_j\rangle\langle \psi_j| \, 
\end{equation}
where  $\{ |\psi_j \rangle \}_{j=1}^{2^{d}}$ comprise a basis in the Hilbert space of the quantum computer $\bbC^{2^d}$. We also assume the observable $O$ that we measure is diagonal in the same basis i.e.
\begin{equation}\label{eq:observable}
O =  \sum_{j=1}^{2^d} o_j |\psi_j\rangle\langle \psi_j| \, 
\end{equation}
where $\{o_j\}_{j=1}^{2^d}$ are the eigenvalues of $O$ and $p_j = \langle \psi_j | \rho |\psi_j \rangle$ 
is the probability of observing state a measurement outcome corresponding to the eigenvalue $o_j$.

Without loss of generality we can interpret the probability distribution $p$ as a classical probability density $p(x)$ over the binary random variable $x \in \{0,1\}^d$. 
The mathematical task at hand is then to find $p(x)$ consistent with the observed outcomes $\{o_j\}$. A canonical example of such density in our context is the Ising model
\begin{equation}\label{eq:ising}
p(x; A) = \frac1{Z(A)} \exp( x^T A x) \, , 
\quad Z(A) = \sum_{x \in \{0,1\}^d} \exp( x^T A x) \, , 
\end{equation}
which is parametrized by some matrix $A\in \bbR^{d\times d}$.


Given first and second moments of a distribution $\Pi$ on  $\{0,1\}^d$,  
the 
MaxEnt distribution  is 
the Ising model \eqref{eq:ising} which satisfies the moment constraints \cite{cover2012elements} .
With finitely many observations, one can replace the moments with 
observed empirical moments, and the corresponding MaxEnt distribution 
is the same one which maximizes the likelihood under a 
MaxEnt 
ansatz for its parameterized form \cite{murphy2012machine}.
Through a Bayesian framework \cite{green2015bayesian}, 
this likelihood can be combined with a prior on the parameters, 
resulting in a posterior distribution, which quantifies the uncertainty in 
the distribution given the observations.
We will develop and implement novel algorithms for the solution of these problems
with a cost which is polynomial in $d$,
based on recently developed technologies mentioned above.
Now we proceed with a more rigorous mathematical setup, 
and a summary of the results of the paper.


\section{Mathematical Setup and Summary of Results}\label{sec:mathsetup}

Let $E=\{0,1\}$, and denote the state $x = [x_1,\dots, x_d] \in E^d$.  
Consider a  probability distribution $\Pi : \sigma(E^d) \times \Xi \rightarrow [0,1]$, parameterized by $\lambda \in \Xi$, 
with density $\pi:E^d\times \Xi \rightarrow [0,1]$.
The target $\Pi(\cdot|\lambda)$ represents the distribution over $d$ dependent qubits $X=[X_1, \dots, X_d]^T$ 
given a particular value of $\lambda$, and the objective of this work is to identify $\lambda$.
We will denote the un-normalized measure by $Q: \sigma(E^d) \times \Xi \rightarrow [0,1]$, with density $q:E^d\times \Xi \rightarrow [0,1]$,
so that $\Pi = \frac1{Z(\lambda)} Q$ and $\pi=\frac1{Z(\lambda)} q$, 
with $Z(\lambda)=Q(1|\lambda) := 
\sum_{x\in E^d} 1 q(x|\lambda)$.

Let $\phi_j: E^d \rightarrow \bbR$, for $j=1,\dots, J$, and define $\phi = [\phi_1,\dots, \phi_J]^T: E^d \rightarrow \bbR^J$.
It is assumed that, for a particular unknown value of $\lambda$, we are given either 
\begin{probP}\label{prob1}
$J$ exact moments $m = [m_1, \dots, m_J]^T$:
\begin{equation}\label{eq:moms}
m = \Pi(\phi | \lambda) 
:= 
\sum_{x\in E^d}  \phi (x) \pi(x|\lambda) \, , 
\end{equation}
\end{probP}
\begin{probP}\label{prob2}
$M$ independent and identically distributed (i.i.d.) 
noisy observations of this process $y^{(i)} \sim \Pi(\cdot|\lambda)$, for $i=1,\dots, M$. 
\end{probP}
As mentioned above, the aim in either case is to reconstruct $\Pi$ by identifying $\lambda$.

First, consider problem (P\ref{prob1}).
Assume we seek the probability distribution $p$ which maximizes the entropy 
\begin{equation}\label{eq:entropy}
H(p) = - \sum_{x\in E^d} p(x) \log(p(x)) \, ,
\end{equation}
subject to the constraints \eqref{eq:moms}, with $p$ replacing $\pi(\cdot|\lambda)$.
This is in fact a convex optimization problem. Therefore, if there exists a probability distribution $p$
which satisfies the constraints \eqref{eq:moms}, 
then the necessary first order optimality conditions are also sufficient for global optimality \cite{boyd}.
For this particular problem, the first order conditions can be simplified to 
$p(x)=\pi(x|\lambda) = \frac1{Z(\lambda)}q(x|\lambda)$, where
\begin{equation}\label{eq:maxent1}
q(x|\lambda) = \exp(\lambda^T \phi(x)) 
\, ,
\end{equation}
and the vector of Lagrange multipliers $\lambda=[\lambda_1,\dots,\lambda_J]^T$ 
is the solution to the system of equations 
\begin{equation}\label{eq:lambda}
\Pi(m-\phi | \lambda) = \frac{Q(m-\phi|\lambda)}{Q(1|\lambda)} = 0 \, .
\end{equation}
This is referred to as the method of \emph{maximum entropy} (MaxEnt).
See Chapter 11.1 of \cite{cover2012elements} for the derivation.
A density which is parametrized in the form \eqref{eq:maxent1} is more generally referred to as a
MaxEnt ansatz \cite{murphy2012machine}. 

In general, the system of equations \eqref{eq:lambda} does not have a closed form solution,
and it is 
a challenging problem to solve, particularly for large $d$, 
as one evaluation of the system would consist of $\cO(2^d)$ operations.  
The issue is known as the curse of dimensionality, and precludes the use of standard iterative algorithms for solving the problem.
We propose to utilize a Robbins Monro type \cite{robbins1951stochastic, kushner2003stochastic} stochastic iterative algorithm for its solution,
together with a sequential Monte Carlo (SMC) sampler \cite{delm:06b} 
\emph{unbiased estimator} of a system of equations equivalent to \eqref{eq:lambda}, $Q(m-\phi|\lambda)=0$,
within each iteration of the Robbins Monro. 
It is well-known that SMC samplers are stable with $\cO(d^2)$ 
operations for such targets which can be evaluated pointwise (in $x$)  \cite{beskos2014stability},
thus bypassing the curse of dimensionality.
That is the first major contribution of the paper.

Now, consider problem (P\ref{prob2}),
concerned with recovering $\pi(\cdot|\lambda)$ from the i.i.d. observations $Y:=\{y^{(i)}\}_{i=1}^M$.
Following the MaxEnt ethos, we choose an appropriate set of functions $\phi$ as above \eqref{eq:moms},  
based on our assumptions of the underlying quantum system, 
and adopt the ansatz \eqref{eq:maxent1} for the \emph{likelihood} of a single observation, given $\lambda$: 
$\bbP(y^{(i)}|\lambda) = \pi(y^{(i)} | \lambda)$.
For example, one could use the set of first and second moments 
$\phi(x)={\rm vec}(\{x_i x_j\}_{1\leq  i \leq j \leq d})$.  
The likelihood of 
$Y:=\{y^{(i)}\}_{i=1}^M$ given $\lambda$ is given by
\begin{equation}\label{eq:like}
\bbP(Y|\lambda) = \prod_{i=1}^M \pi(y^{(i)} | \lambda) \, .
\end{equation}
The maximizer of \eqref{eq:like} is called the 
maximum likelihood estimator (MLE). It is equivalent to the maximizer of 
$\log \bbP( Y| \lambda)$. We have 
\begin{equation}\label{eq:loglike}
\nabla_\lambda \log \bbP( Y| \lambda) =  M \Pi( \widehat m - \phi | \lambda) \, ,
\end{equation}
where 
\begin{equation}\label{eq:empirical}
\widehat m =\frac1M \sum_{i=1}^M \phi(y^{(i)}) \, .
\end{equation}
The necessary first order optimality condition is vanishing gradient, i.e.
\begin{equation}\label{eq:lambda-empirical}
\Pi(\widehat m - \phi | \lambda) = 0 \, ,
\end{equation}
which is clearly equivalent to 
to the MaxEnt solution \eqref{eq:lambda}, except with empirical moments replacing exact moments. 
This problem can therefore be solved using the same method described above.

This solution is still somewhat unsettling. Even if we are content with our choice of $\phi$, 
we have arrived at a point estimate for inference given a finite set of observations. 
The scrupulous Bayesian will not allow the story to end here. 
Indeed for the sake of uncertainty quantification (UQ), we require a \emph{posterior distribution} 
over $\pi(\cdot | \lambda)$, or equivalently, a posterior distribution over $\lambda$.
In order to accomplish this, we complete the picture with a prior on $\lambda$, giving 
\begin{equation}\label{eq:posterior}
\bbP(\lambda | Y) \propto \bbP( Y| \lambda) \bbP(\lambda) \, ,
\end{equation}
where the constant of proportionality depends only upon the observations $Y$.
The second major contribution of the paper concerns quantifying a posteriori uncertainty using this approach.
Due to the fact that the likelihood $\bbP( Y| \lambda)$ itself has an intractable normalizing constant 
$Q(1|\lambda)^M$, this notoriously difficult problem is referred to as {\em doubly-intractable} \cite{moller2006efficient, murray2006mcmc}. 

If we could at least evaluate a non-negative and unbiased estimator of the likelihood $\bbP( Y| \lambda)$, 
or equivalently \eqref{eq:posterior}, then there are algorithms which can be used, 
such as  pseudo-marginal Markov chain Monte Carlo (MCMC) method \cite{beaumont2003estimation, andrieu2009pseudo}. 
The recent work \cite{lyne2015russian} proposes to construct unbiased but signed estimates 
of the target \eqref{eq:posterior}, 
and then they use such estimate within pseudo-marginal MCMC by considering the absolute value of the estimate 
in the acceptance ratio and then storing the sign of the accepted samples in order to correct the ultimate estimator.
The work \cite{wei2017markov} proposes to use alternative unbiased estimators in this same framework.
Note that earlier work also appeared which
allows one to sample from such a target \cite{moller2006efficient, murray2006mcmc},
however these works require being able to sample from $\pi(\cdot|\lambda)$ for a given value of $\lambda$. 
This is possible in some cases, for example 
the Ising model (or Boltzmann machine) \cite{propp1996exact, childs2001exact}, but not in general.  
The works \cite{lyne2015russian, wei2017markov} are generally applicable.

We will take a new approach to doubly intractable simulation in the present work.
A number of recently developed methods require only unbiased estimators of \eqref{eq:loglike}.
These include the stochastic gradient Langevin dynamics (SGLD) \cite{teh2011bayesian},
the zig-zag sampler \cite{bierkens2016zig}, or the bouncy particle sampler \cite{bouchard2018bouncy, peters2012rejection}.
We propose to construct an 
unbiased estimator of \eqref{eq:loglike} using the debiasing technique 
from the recent works \cite{mcleish2011general,rhee2015unbiased}, 
and then use this unbiased estimator within such MCMC methods.
This is the second major contribution of this work.
In particular, we constrain our attention here to the SGLD.  

The rest of the paper is organized as follows. In Section \ref{sec:smcsamp}
the SMC sampler is introduced, which will be used as a crucial component within all subsequent algorithms.
In particular, it will be shown how to construct consistent estimators of $\Pi(\varphi | \lambda)$ and unbiased estimators of $Q(\varphi | \lambda)$.
In Section \ref{sec:sa} we introduce the Robbins Monro algorithm, and present the first main result regarding 
its implementation using the SMC sampler. 
Section \ref{sec:uq} concerns Bayesian posterior UQ in the context of finitely many observations.
In Section \ref{sec:simos} we simulate data from a toy model, and illustrate the algorithm from Section \ref{sec:sa}
for computing the MaxEnt solution using both exact moments and observed moments (where it coincides with the MLE), 
as well as the algorithm from Section \ref{sec:uq} for UQ in the case of observations.


%
%

\section{SMC samplers for estimation of integrals}
\label{sec:smcsamp}

Here it will be described how to estimate expectations 
$\Pi(\varphi | \lambda) := 
\sum_{x\in E^d} \varphi(x) \pi(x|\lambda)$,
for bounded functions $\varphi:E^d \rightarrow \bbR$, 
and construct non-negative unbiased estimators 
of $Q(\varphi | \lambda)=
\sum_{x\in E^d} \varphi(x) q(x|\lambda)$, 
for a given value of $\lambda$, 
using an auxillary variable technique. 

Let $\eta_0 = \Gamma_0$ be a known probability distribution over 
a general state space $U$, 
and let $\eta_j$ be some probability over $U$ for $j=0,\dots,J$, 
whose density can be evaluated up to a normalizing constant, 
but cannot be sampled from directly.
Define the unnormalized measure associated to $\eta_j$ 
by $\Gamma_j$, with density $\gamma_j$, so that 
\begin{equation}\label{eq:targ}
\eta_j(dx) = \Gamma_j(dx)/\Gamma_j(1) \, ,
\end{equation}
where $\Gamma_j(1) = 
\int_{U} \gamma_j(x) dx$.

The objective is to 
approximate expectations for bounded functions $\varphi:U \rightarrow \bbR$
$$
\eta_J(\varphi) := \int_{U} \varphi(x) \eta_J(dx) \, . 
$$
We cannot sample from this distribution, but we can obtain a convergent estimator 
$$
\eta_J(\varphi) \approx \frac1{N} \sum_{i=1}^N \varphi(x^i_J) \, ,
$$
using SMC samplers,
as described below \cite{delm:06b} (see also \cite{jarzynski1997nonequilibrium, neal2001annealed, chopin2002sequential}).


\begin{rem} 
In the context of this paper 
the target $\eta_J = \Pi(\cdot|\lambda)$ will be considered in each inner iteration 
of the stochastic iterative algorithm. 
So $\gamma_J(x) = q(x|\lambda)$. 
To make the presentation more concrete, let $\gamma_0 \propto 1$ and 
\begin{equation}\label{eq:smc_intermediate}
\gamma_j(x) = 
\exp\left [ \sum_{i=1}^j \lambda_i \phi_i(x) \right ] \, .
\end{equation}
However, note that the choice of intermediate targets $\gamma_j$ to interpolate between an initial density $\gamma_0$, 
which can be sampled from exactly, and $\gamma_J$ are essentially arbitrary.
For example, one could incorporate several constraints per iteration or anneal in between individual constraints.
\end{rem}

\subsection{Estimating the target}
\label{sec:smc}


Define 
$$
G_j(x) = \gamma_{j+1}(x)/\gamma_j(x) \, .
$$
Let $M_j$ denote a Markov kernel such that 
$$
(\eta_j M_j)(dx) := \int_{U} \eta_j(dx') M_j(x', dx) =  \eta_j(dx) \, .
$$

Iterate Algorithm \ref{algo:smc} and define 
\begin{equation}\label{eq:empirical}
\eta_j^N(\varphi) := \frac1{N} \sum_{i=1}^N \varphi(x_j^{i}) \, .
\end{equation}

\begin{algorithm}[h]
\caption{SMC sampler} 
Let $x_0^{i} \sim \eta_0$.  For $j=1,\dots,J$, repeat the following steps for $i=1,\dots, N$:
\begin{itemize}
\item[(i)] Define 
$w_j^i = G_{j-1}(x_{j-1}^{i}) / \sum_{k=1}^N G_{j-1}(x_{j-1}^{k})$.
\item[(ii)] Resample. e.g. select $I_j^i \sim \{ w_j^1, \dots, w_j^N\}$, and let $\hat x_j^{i} = x_{j-1}^{I_j^i}$.
\item[(iii)] Mutate. Draw $x_j^{i} \sim M_j(\hat x_j^{i},\cdot)$.
\end{itemize}\label{algo:smc}
\end{algorithm}

\begin{ass} \label{ass:G}
There exists some $C > 0$ such that for all 
$x\in U$
$$
0<G_j(x) \leq C \, .
$$
\end{ass}

\begin{prop}
Assume \ref{ass:G} holds. 
Then for $N>1$, $p\geq 1$,
$$
\bbE | \eta_j^N(\varphi) - \eta_j(\varphi) |^p \leq \frac{C(j)}{N^{p/2}} \, .
$$
\end{prop}

The constant $C(j)$ can be made to be independent of $j$ if an additional strong mixing condition is assumed on $M_j$, uniformly in $j$ \cite{delm:04}.

\subsection{Estimating the normalizing constant}
\label{sec:noco}

Recall $\Gamma_j(1) = \sum_{x\in U} \gamma_j(x)$, and observe that 
$$
\eta_j(G_j) = \frac{1}{\Gamma_j(1)} \sum_{x\in U} \frac{\gamma_{j+1}(x)}{\gamma_{j}(x)} \gamma_{j}(x)  = \Gamma_{j+1}(1)/\Gamma_j(1) \, .
$$ 
Noting that $\Gamma_0(1)=\eta_0(1)=1$, therefore
$$
\Gamma_J(1) = \prod_{j=0}^{J-1} \eta_{j}(G_{j})\, .
$$
This is the normalizing constant of $\eta_J$. 

Observe that using Algorithm \ref{algo:smc} 
we can construct 
an estimator by 
$$
\eta^N_{j}(G_{j}) = \frac1N \sum_{i=1}^N G_j(x_j^{i}) \, . 
$$
Recall that for any $\varphi:U \rightarrow \bbR$ we have  
$\Gamma_J(\varphi) = \Gamma_J(1) \eta_J(\varphi)$ by definition \eqref{eq:targ}.  
Define 
\begin{equation}\label{eq:unbiased}
 \Gamma_J^N(\varphi) := \prod_{j=0}^{J-1} \eta^N_{j}(G_{j}) \eta_J^N(\varphi) \, . 
\end{equation}
The following proposition is proven in \cite{delm:04}.
The proof is reproduced in Appendix \ref{app:1} for completeness, due to its importance to the results of this paper.

\begin{prop}\label{prop:unbiasedunno}
Assume \ref{ass:G} holds. 
Then $\bbE \Gamma_J^N(\varphi) = \Gamma_J(\varphi)$. 
\end{prop}

\section{Stochastic approximation}
\label{sec:sa}

Here it will be assumed that we are given values 
$m = \Pi(\phi|\lambda) \in \bbR^J$, i.e. Problem (P\ref{prob1}). 
As mentioned in Section \ref{sec:mathsetup}, the MaxEnt approach entails solving the convex optimization problem 
\eqref{eq:entropy} (which has domain $\{x; p(x) > 0\}$) subject to the constraints 
$\sum_{x\in E^d} p(x) = 1$, and \eqref{eq:moms}.
It is necessary and sufficient to solve the equation \eqref{eq:lambda}, reproduced below

\begin{equation}\label{eq:maxent_fp}
\frac{Q(m - \phi | \lambda)}{Q(1|\lambda)} = 0 \, .
\end{equation}

As mentioned above, a naive computation of the left-hand side of the above equation for a single value of $\lambda$
would incur a cost of $\cO(2^d)$, which makes the problem intractable for very small values of $d > 15$ or so.
It is however possible to obtain a noisy estimate 
using a stochastic simulation algorithm.
Algorithms such as MCMC \cite{hastings1970monte, gilks2005markov} are designed to sample the 
state space in $\cO(d)$ steps \cite{roberts2001optimal},  
and provide a consistent estimate of \eqref{eq:maxent_fp} for a given value of $\lambda$.
 
Returning to \eqref{eq:maxent_fp}, observe that solving this equation is equivalent to solving $Q(m-\phi|\lambda) =0$.
Assume we can simulate unbiased estimates $\widehat F (\lambda)$, 
i.e. estimates such that $\bbE \widehat F (\lambda) = Q(m-\phi | \lambda)$.
The Robbins Monro algorithm \cite{robbins1951stochastic} for stochastic fixed point iteration for 
finding a root of the equation $Q(m-\phi | \lambda)=0$ takes the form
\begin{equation}\label{eq:robbins_monro}
\lambda^{k+1} = \lambda^k + \delta_k \widehat{F}(\lambda^k) \, .
\end{equation}
Under technical conditions, and for appropriate choice of $\delta_k$ such that $\sum_k \delta_k =\infty$ and $\sum_k \delta_k^2 <\infty$,
it can be shown that this algorithm converges to a fixed point of \eqref{eq:maxent_fp} \cite{robbins1951stochastic, kushner2003stochastic}
without ever solving the equation exactly, even once.

We make the following assumption here, which is sufficient to ensure Assumption \ref{ass:G}

\begin{ass}\label{ass:q}
There exists a $c > 0$ such that for all $\lambda \in \Xi$ and $x\in E^d$
$$
c < q(x|\lambda) < c^{-1} \, .
$$
\end{ass}

\begin{prop}\label{prop:robbinsmonro}
Assume \ref{ass:q}.
Let $\widehat F (\lambda)$ be such that $\bbE \widehat F (\lambda) = Q(m-\phi | \lambda)$. 
If $\{\delta_k\}_{k=1}^\infty$ is chosen so that $\sum_k \delta_k =\infty$ and $\sum_k \delta_k^2 <\infty$, 
then $\lambda^k \rightarrow \lambda^*$, where $\lambda^k$ is given by \eqref{eq:robbins_monro} and 
$\lambda^*$ is the solution to $Q(m-\phi|\lambda^*)=0$, 
hence \eqref{eq:maxent_fp}.
\end{prop}

Consider now running the SMC sampler Algorithm \ref{algo:smc} with the choice 
$\gamma_0 \propto 1$ and $\gamma_j = \exp\{\sum_{i=1}^j \lambda_i \phi_i\}$
to construct the unbiased estimator \eqref{eq:unbiased} $\Gamma_J^N(m-\phi | \lambda)$
such that $\bbE \Gamma_J^N(m-\phi | \lambda) = Q(m-\phi | \lambda)$.
With this unbiased estimator of $Q(m-\phi | \lambda)$, 
we can use Robbins Monro to solve the fixed point problem for $\lambda$ \eqref{eq:maxent_fp}, 
hence the MaxEnt target $\pi(x|\lambda)$.  
The following theorem makes this statement precise.

\begin{theorem}[Stochastic approximation with SMC] \label{thm:sasmc}
Assume \ref{ass:q} holds, and let $\lambda^k$ be generated by the iteration \eqref{eq:robbins_monro},
where $\widehat{F}(\lambda) := \Gamma_J^N(m - \phi | \lambda)$
is independently generated from an SMC sampler algorithm at each iteration, as defined in \eqref{eq:unbiased}. 
If $\{\delta_k\}_{k=1}^\infty$ is chosen so that $\sum_k \delta_k =\infty$ and $\sum_k \delta_k^2 <\infty$
then $\lambda^k \rightarrow \lambda^*$, where $\lambda^*$ is a solution of 
$Q(m-\phi|\lambda^*)=0$, hence \eqref{eq:maxent_fp}.
\end{theorem}

\begin{proof}
By Proposition \ref{prop:unbiasedunno} $\bbE\widehat{F}(\lambda) = Q(m-\phi|\lambda)$.
The conclusion is an immediate consequence of Proposition \ref{prop:robbinsmonro}.
\end{proof}

Assume the evaluation of $\phi_j$ is $\cO(1)$. Then the evaluation of $\gamma_j$ is $\cO(j)$, 
and since there will be $J$ steps, the total cost of the algorithm will be $\cO(J^2)$.
Since $\lambda \in \bbR^{J}$, one may require up to $\cO(J)$ steps of Robbins Monro.
Therefore, a conservative estimate on the cost is 
$\cO(J^3)$.  Recall that $J$ is the number of constraints. 
For $p^{\rm th}$ moments, there are $J={{d+p-1}\choose{p}} \lesssim d^p$ constraints, 
yielding a conservative estimate on the final cost of $\cO(d^{3p})$. 
The constant should be rather small however, and this is of course a vast improvement on $\cO(2^{d})$
for even the best linear algorithm for solution with exact evaluation of \eqref{eq:maxent_fp}.

\section{Uncertainty Quantification}
\label{sec:uq}

In the previous section we arrive at a single optimal value $\lambda^*$, and a single MaxEnt target likelihood, 
i.e. uncertainty is not quantified.  Here we venture to quantify uncertainty by 
putting a prior on $\lambda$, 
and computing (sampling from) the posterior $\bbP(\lambda|Y)$ \eqref{eq:posterior}.  
This leads to a so-called {doubly-intractable} target, as 
described above.  
In this section a new method is introduced to sample from the target.
In particular, in subsection \ref{ssec:debias} we introduce the debiasing method mentioned above, 
and describe how to construct unbiased (signed)
estimators of $\Pi(\varphi | \lambda)$ using estimator(s) from the SMC sampler.
In subsection \ref{ssec:sgld} we introduce the SGLD algorithm and present the second main result 
of the paper, regarding its implementation using the debiased estimators. 
The section is concluded
with a description in subsection \ref{sec:error} of how to separate an explicit error distribution, 
which is more expensive and is left to future work.

\subsection{Debiasing}
\label{ssec:debias}

Some MCMC 
methods, such as SGLD 
\cite{teh2011bayesian}, 
as well as some recently developed piecewise-deterministic Markov processes (PMDP) 
\cite{davis1984piecewise, bierkens2016zig, bouchard2018bouncy, peters2012rejection},
require only an unbiased estimator of the gradient of the log-likelihood for implementation.  
Others, such as pseudo-marginal MCMC, require an unbiased and non-negative estimator of the likelihood itself. 
It is generally much easier to construct unbiased estimators which are signed \cite{mcleish2011general, rhee2015unbiased}. 
The work \cite{lyne2015russian} shows how one can utilize the pseudo-marginal MCMC even with a signed estimator.  
However, one may easily argue that it is more natural to use 
a signed and unbiased estimator within SGLD or one of the PDMP algorithms.
We consider the latter strategy here.

Suppose that we have a collection of 
random variables $\{\Delta_l\}_{l=0}^\infty$, where 
\begin{equation}\label{eq:fin}
\sum_{l=0}^\infty \bbE | \Delta_l | 
<\infty,
\end{equation}
and we are interested in estimating $\bbE Z$, where 
\begin{equation}\label{eq:zsum}
\bbE Z = \sum_{l=0}^\infty \bbE \Delta_l = \sum_{l=0}^\infty p_l \frac{\bbE \Delta_l}{p_l} \, .
\end{equation}
Note that \eqref{eq:fin} guarantees that $\bbE Z < \infty$.
If we draw $L \sim {\bf p}$, where ${\bf p} = [p_0,\dots]$, then it is clear that 
\begin{equation}\label{eq:singleton}
\widehat{Z}_s = \frac{\Delta_L}{p_L}
\end{equation} 
is an unbiased estimator of $\bbE Z$.
Furthermore, if we instead draw $L \sim {\bf \tilde p}$, where ${\bf \tilde p}$ is defined such that
$p_\ell = \sum_{l\geq \ell} \tilde{p}_l$, then one can see that 
\begin{equation}\label{eq:multi}
\widehat{Z}_t = \sum_{l=0}^L \frac{\Delta_l}{p_l} 
\end{equation}
is again an unbiased estimator:
\begin{equation}\label{eq:telescope}
\bbE \sum_{l=0}^L \frac{\Delta_l}{p_l} = \bbE \sum_{l=0}^\infty 1_{L\geq l} \frac{\Delta_l}{p_l}  = \sum_{l=0}^\infty \bbE(L\geq l) \frac{ \bbE \Delta_l}{p_l}  =
\sum_{l=0}^\infty p_l \frac{ \bbE \Delta_l}{p_l}  \, .
\end{equation}
The exchange of the infinite sum and expectation is allowed by \eqref{eq:fin} and Fubini theorem.

\begin{ass}\label{ass:rgvar}
Let \eqref{eq:fin} hold and assume the probability ${\bf p}$ is chosen such that $\sum_{l=0}^\infty \frac{\bbE \Delta_l^2}{p_l} < \infty$.
\end{ass}

\begin{prop}\label{prop:rgconverge}
For either estimator $r\in \{s,t\}$, let $\widehat Z^i_r$ be i.i.d. draws of the estimator \eqref{eq:singleton} or \eqref{eq:multi}, 
for $i=1,\dots, N$.
Under assumption \ref{ass:rgvar}, $\bbE \widehat Z_r^i = \bbE Z$ and we have, 
$$
\sqrt{N} \left( \frac1N \sum_{i=1}^N \widehat Z_r^{i} - \bbE Z \right) \rightarrow N(0, \Sigma) \, ,
$$ 
with $\Sigma = \sum_{l=0}^\infty \frac{\bbE \Delta_l^2}{p_l} < \infty$. 
The expected cost is given by 
$$
\bbE ({\rm Cost}) = \sum_{l=0}^\infty C(\Delta_l) p_l \, ,
$$
where $C(\Delta_l)$ is the cost to obtain the realization $\Delta_l$.
\end{prop}

See \cite{mcleish2011general, rhee2015unbiased, lyne2015russian, vihola2017unbiased} and references therein for proof and further discussion of such approaches. 
In particular, this approach will allow us to transform consistent estimators to unbiased estimators for use in the algorithms mentioned above.
The recent literature on PDMP is quite extensive, and so we defer consideration of these algorithms to future work 
and focus on SGLD, which will be described below in generality.
First we illustrate how to construct such de-biased estimator of \eqref{eq:loglike}.


\subsubsection{Debiased SMC sampler estimator}
\label{sssec:smc_debias}

We now describe one way to construct an unbiased estimator of \eqref{eq:loglike} directly, 
using this de-biasing trick, without multiplying through by $Q(1|\lambda)$.
In particular, the single term estimator \eqref{eq:singleton} 
will be considered. 

Let $N_l=N_0 2^{2l}$. 
Replace step (ii) from the SMC samplers algorithm with $\hat x_j^{i} = x_{j-1}^i$, i.e. 
omit the resampling step. This algorithm was named annealed importance sampling (AIS) in \cite{neal2001annealed}.
The estimator $\frac1N \sum_{i=1}^N \prod_{j=1}^J G_{j-1}(x_{j-1}^i) \varphi(x_J^i)$ 
is an unbiased estimator of $Q(\varphi|\lambda)$,
and the individual terms are independent: $x_j^i \perp x_j^{i'}$
for $i\neq i'$ and all $j=1,\dots,J$. 
This means in particular that if we construct 
\begin{equation}\label{eq:ais}
\tilde{\eta}_J^{N/4}(\varphi) := \frac{\sum_{i=1}^{N/4} \prod_{j=1}^J G_{j-1}(x_{j-1}^i) \varphi(x_J^i)}{\sum_{i=1}^{N/4} \prod_{j=1}^J G_{j-1}(x_{j-1}^i)} \, ,
\end{equation}
from two independent replications,
one with $N$ total particles and one with $N/4$ total particles, then two estimators have the same distribution;
in particular, the same expectation. 
So we have the following proposition.

\begin{prop}\label{prop:unbiasedpi}
Define the coupled estimator $\Delta_l := \tilde \eta^{N_l}_J(\varphi) - \tilde \eta^{N_{l}/4}_J(\varphi)$,
with particles coming from the same realization of the algorithm described above with $N_l$ particles. 
Then $\bbE \widehat{Z}_s = \eta_J(\varphi)$, with $\widehat{Z}_s$ given by \eqref{eq:singleton}. 
The factor $4$ here relates to the specific choice of $N_l$ and can be replaced with $N_{l+1}/N_l$ in general.
\end{prop}

\begin{proof}
Following from the explanation above, the equation \eqref{eq:zsum} takes the form 
$\eta_J(\varphi) = \sum_{l=1}^\infty \bbE \Delta_l$,
so the result follows from Proposition \ref{prop:rgconverge}. 
\end{proof}

Observing that in this case $\bbE \Delta_l^2 \propto N_l^{-1}$ and $C(\Delta_l) \propto N_l$, we can observe
that this is the so-called sub-canonical case \cite{rhee2015unbiased}. That is, if the rate of convergence of $\bbE \Delta_l^2$ were faster
(or rate of growth of expected cost slower), then we could choose $p_l$ in such a way that both asymptotic variance and cost 
were finite.  As this is not the case, we instead choose $p_l \propto N_l^{-1} l \log(l)$, 
as proposed in \cite{rhee2015unbiased} (see also \cite{franks2018unbiased}), 
which results in finite variance and expected cost which is finite with high probability.

\subsection{SGLD}
\label{ssec:sgld}

Given a symmetric positive definite matrix $P(\lambda)\in \bbR^{J\times J}$ 
the following SDE keeps $p(\lambda)$ invariant \cite{girolami2011riemann, green2015bayesian}
\begin{equation}\label{eq:ctslang}
d\lambda = g(\lambda) dt + P(\lambda)^{1/2} dW 
\end{equation}
where for $j=1, \dots, J$
\begin{equation}\label{eq:gee}
g(\lambda)_j = [ P(\lambda) \nabla \log p(\lambda) ]_j + 
\sum_{k=1}^J \frac{ \partial P_{jk} }{\partial \lambda_k}(\lambda) \, .
\end{equation}

For constant $P(\lambda)$, 
it was proposed in \cite{teh2011bayesian} to replace $\nabla \log p(\lambda)$ 
by an unbiased estimator $\widehat{\nabla \log p(\lambda)}$ 
(i.e. $\bbE \widehat{\nabla \log p(\lambda)} = \nabla \log p(\lambda)$) and 
approximate the SDE using Euler Maruyama method with stepsize $\delta_n \propto 1/n$, similar to Robbins Monro, as follows  
\begin{eqnarray}\label{eq:rmsgld}
\lambda^{n+1} = \lambda^n + 
\frac {\delta_n}{2} \left( P(\lambda^n) \widehat{\nabla \log p(\lambda^n)} + \sum_{k=1}^J \frac{ \partial P_{jk} }{\partial \lambda_k^n}(\lambda^n) \right) 
+ N(0,\delta_n P(\lambda^n)) \, .
\end{eqnarray}

For the moment, assume $P(\lambda)=1$, 
and $\widehat{\nabla \log p(\lambda)}=\nabla \log p(\lambda)$, 
and define the measure corresponding to $p(\lambda)$ by $\mu$. 
Furthermore, define the empirical measure corresponding to $K$ steps of the 
constant $\delta$-discretized version of \eqref{eq:rmsgld} by 
$\mu^K_\delta(\varphi) = \frac1K \sum_{n=1}^K \varphi(\lambda^n)$ and
define the limiting diffusion \eqref{eq:ctslang} sampled at the same $K$ points $i\delta$, $i=1,\dots, K$, by $\mu^K$.
Then we have
\begin{equation}\label{eq:triangle}
\bbE |  \mu^K_\delta(\varphi) - \mu(\varphi) |^2 \leq 
2 \underbrace{\bbE   | \mu^K_\delta(\varphi) - \mu^K(\varphi) |^2}_{\propto~ \delta^2 ~ {\rm - discretization~error}} + 
2 \underbrace{\bbE   | \mu^K(\varphi) - \mu(\varphi) |^2}_{\propto~ (\delta K)^{-1} ~ {\rm - sample~error}} \, .
\end{equation}
The discretization error $\delta^2$ comes from the standard Euler-Maruyama strong convergence rate for SDE with constant diffusion.
Define the auto-covariance 
$$\rho_n := \bbE(\lambda^n_\infty - \bar{\lambda}_\infty)( \lambda^0_\infty - \bar{\lambda}_\infty) \, ,$$
where $\lambda_\infty^n$ is a sample from \eqref{eq:ctslang} at time $\delta n$ and $\bar{\lambda}_\infty = \int \lambda \mu(d\lambda)$.
One expects $\rho_n \approx (1-\delta)^n$. 
Therefore the integrated auto-covariance 
IACT$=2\sum_{n=0}^\infty \rho_n \approx \delta^{-1}$, and
the sample error is asymptotically (after burn-in) proportional to 
IACT$/K \approx (\delta K)^{-1}$.
Balancing the error terms, one finds that $\delta \propto K^{-1/3}$ is optimal, leading to MSE of $K^{-2/3}$.
It is slightly worse than the MSE of $K^{-1}$ for $K$ steps of standard MCMC methods.  
This simple argument suggests that the optimal choice for $\delta_n$ is $\delta_n \propto n^{-1/3}$, which is not the same as for SGD.
In fact, this result is made rigorous in the work \cite{teh2016consistency, vollmer2016exploration} (see also references therein), 
leading to the following proposition.


\begin{prop}\label{prop:sgldconv}
Let $P(\lambda)=1$, define the measure corresponding to $p(\lambda)$ by $\mu$, 
choose $\delta_n \propto n^{-1/3}$, and let $\lambda^n$ be generated by \eqref{eq:rmsgld}. 
Then for suitably regular $\varphi$, one has
\begin{equation}
\frac{\sum_{n=1}^K \delta_n \varphi(\lambda^n)}{\sum_{n=1}^K \delta_n} \rightarrow \mu(\varphi) \, . 
\label{eq:sgldconv}
\end{equation}
Furthermore, this choice of $\delta_n$ is asymptotically optimal in terms of cost, 
and for $K$ steps, one has MSE $\propto K^{-2/3}$.
\end{prop}

Let $p(\lambda) = \bbP(\lambda|Y) \propto \bbP(Y|\lambda) \bbP(\lambda)$, and recall \eqref{eq:loglike}, 
so that we have
\begin{equation}\label{eq:postgrad}
\nabla \log p(\lambda) = M \Pi(\widehat{m} - \phi | \lambda) + \nabla \log \bbP(\lambda) \, .
\end{equation}
Assume that $\nabla \log \bbP(\lambda)$ can be computed exactly.
Now define an unbiased estimator of \eqref{eq:postgrad} using the method of 
Proposition \ref{prop:unbiasedpi}, i.e. draw $L\sim {\bf p}$, and let 
\begin{equation}\label{eq:unbiased_gradlike}
\widehat{\nabla \log p(\lambda^n)} :=  \frac{M}{p_L} \Delta_L(\widehat{m} - \phi | \lambda) + \nabla \log \bbP(\lambda) \, .
\end{equation}

\begin{theorem}\label{thm:sgld_deb}
Let $P(\lambda)=1$, define the measure corresponding to $p(\lambda)=\bbP(\lambda|Y)$ by $\mu$, 
choose $\delta_n \propto n^{-1/3}$, and let $\lambda^n$ be generated by \eqref{eq:rmsgld}, 
with $\widehat{\nabla \log p(\lambda^n)}$ defined as in \eqref{eq:unbiased_gradlike}.
Then for suitably regular $\varphi$, one has
\begin{equation}
\frac{\sum_{n=1}^K \delta_n \varphi(\lambda^n)}{\sum_{n=1}^K \delta_n} \rightarrow \mu(\varphi) \, . 
\label{eq:sgldconv}
\end{equation}
Furthermore, this choice of $\delta_n$ is asymptotically optimal in terms of cost, 
and for $K$ steps, one has MSE $\propto K^{-2/3}$.
\end{theorem}

\begin{proof}
This follows directly from Propositions \ref{prop:unbiasedpi} and \ref{prop:sgldconv}.
\end{proof}

\subsubsection{Reimann manifold method}
\label{ssec:rmm}

It is assumed in \cite{teh2013stochastic} that Proposition \ref{prop:sgldconv} holds also for deterministic evaluable $P(\lambda)$, 
yielding a Reimann manifold SGLD (RMSGLD), but this is not proven.
In any case, based on the above heuristic one would expect a different scaling of $\delta_n$ to be optimal, since 
the strong rate of convergence is reduced to $\delta$, so balancing terms gives $\delta \propto K^{-1/2}$ and an MSE of $K^{-1/2}$.
Nevertheless, here we further postulate that \eqref{eq:sgldconv} holds indeed for the case in which we have an unbiased estimator 
$\widehat g (\lambda) \approx g(\lambda)$ and a non-negative unbiased estimator $\widehat P (\lambda) \approx P(\lambda)$.


Let $P(\lambda) = Q(1|\lambda) 
= \Gamma_J(1|\lambda)$, 
and let the target be $p(\lambda) = \bbP(\lambda|Y) \propto \bbP(Y|\lambda) \bbP(\lambda)$, so that 
$$
\sum_{k=1}^J \frac{ \partial P_{jk} }{\partial \lambda_k}(\lambda) = 
Q({\bf 1}^T\phi | \lambda)  \, ,
$$
where ${\bf 1} \in \bbR^J$ denotes the vector of all $1$s, and
$$
\nabla \log p (\lambda) = \frac{M Q(\widehat{m} - \phi | \lambda)}{Q(1 | \lambda)}
+ \nabla \log \bbP(\lambda) \, ,
$$ 
where  
$\widehat{m}=\frac1{M} \sum_{k=1}^M \phi(y^{(k)})$. 
Then using the notations of Section \ref{sec:smcsamp}, \eqref{eq:gee} takes the form 
$$
g(\lambda) = \Gamma_J(M(\widehat m - \phi) +{\bf 1}^T\phi | \lambda) + 
\Gamma_J(1|\lambda) \nabla  \log \bbP(\lambda) \, .
$$
This can be unbiasedly approximated by 
\begin{equation}\label{eq:gee_rmsgld}
\widehat{g}(\lambda) = \Gamma^N_J(M(\widehat m - \phi) +{\bf 1}^T\phi | \lambda) 
+ \Gamma_J^N(1|\lambda) \nabla \log  \bbP(\lambda) \, ,
\end{equation}
as can $P(\lambda)$ by simply $\widehat P (\lambda) = \Gamma_J^N(1|\lambda)$. 

This approach would allow direct simulation without de-biasing, since only unbiased estimates of the unnormalized measure $Q$ are required,
which can be constructed directly from SMC samplers, as described in Proposition \ref{prop:unbiasedunno}.
Nonetheless, this is not pursued further here due to the expected higher cost of an optimal implementation.
We note that the cost of these algorithms can be improved using the multilevel Monte Carlo method \cite{giles2016multilevel},
but this is deferred to future work.

\subsection{Separating an explicit error distribution}
\label{sec:error}

Before concluding the section, we mention an extension of the model considered in the rest of the paper which provides deeper UQ.
Suppose that we do not have i.i.d. observations $y^{(i)} \sim \pi(\cdot | \lambda)$, but in fact
we have observations of the form
\begin{equation}\label{eq:error}
d = {\rm mod} ( y + b^\epsilon, 2) \, , \quad  ~ y \sim \pi(\cdot | \lambda) \, ,
\end{equation}
where
$b^\epsilon = [b_1^\epsilon,\dots,b_d^\epsilon]^T, ~ b_i^\epsilon \sim B(\epsilon) ~ {\rm i.i.d.}~$ for  $i=1,\dots, d$, 
and 
$B(\epsilon)$ denotes a Bernoulli with probability $\epsilon$. 
In other words, each $b_i^\epsilon$ is independently $1$ with probability $\epsilon$ and $0$ otherwise.
This model explicitly accounts for error in the observations,
 i.e. with probability $\epsilon$ any qubit may have been observed wrong.
 
We have $\bbP(d|y,\lambda) = \epsilon^{|d-y|}(1-\epsilon)^{1-|d-y|}$.
The likelihood of a single observation $d$ is given by 
$$
\bbP(d | \lambda) = \int_{E^d} \bbP(d|y,\lambda) \pi(y|\lambda) {\rm d}y \, ,
$$
and so we can see that again we have an integral with respect to $\pi(\cdot|\lambda)$.
Let $D=[d^{(1)},\dots,d^{(M)}]$. Due to the independence between observations and errors, we have
$$
\bbP(D|\lambda) = \prod_{i=1}^M \frac{Q(\varphi_{d^{(i)}} | \lambda)}{Q(1 | \lambda)} \, ,
$$
where we defined $\varphi_{d^{(i)}}(y) := 
\epsilon^{|d^{(i)}-y|}(1-\epsilon)^{1-|d^{(i)}-y|}$.

Therefore
$$
\frac1M \nabla_\lambda \log \bbP(D|\lambda) = 
\frac1M \sum_{i=1}^M \frac{Q(\varphi_{d^{(i)}}\phi| \lambda)}{Q(\varphi_{d^{(i)}}|\lambda)}
- \frac{Q(\phi|\lambda)}{Q(1|\lambda)} \, .
$$
One can define
$$
F_D(\lambda) = \frac1M \sum_{i=1}^M Q(1|\lambda) Q(\varphi_{d^{(i)}}\phi| \lambda) \prod_{j=1, j\neq i}^M Q(\varphi_{d^{(j)}}\phi| \lambda) - 
Q(\phi|\lambda) \prod_{j=1}^M Q(\varphi_{d^{(j)}}\phi| \lambda) \, .
$$
An unbiased estimator of this function can be constructed from M independent SMC samplers, one for each factor.
Despite being massively parallel, this is a significant computational burden, particularly for large $M$. 
Investigation of efficient algorithms for solving this problem is the topic of ongoing research.

\section{Simulations}\label{sec:simos}
\label{sec:simos}

\subsection{Maximum likelihood: exact marginals}
\label{ssec:unbias_rm}

Recall the form of $\widehat F (\lambda)$, for a fixed $\lambda$ from \eqref{thm:sasmc}. 
Using the machinery from subsection \ref{sec:noco} we could construct such an unbiased estimator,
for any sequence of intermediate targets $\gamma_j(x|\lambda)$, and any $J$, provided that 
$\gamma_J(x|\lambda) = q(x|\lambda)$.  
The natural choice is of course to let
\begin{equation}\label{eq:smc_intermediate}
\gamma_j(x|\lambda) = \exp\left [ \sum_{i=1}^j \lambda_i \phi_i(x) \right ] \, ,
\end{equation}
as above, 
and this choice will be used here.
The MCMC kernels $M_j$ are chosen as Metropolis-Hastings kernels with the reversible random walk proposal
$q(\lambda,\lambda')$ defined by $\lambda'_i = {\rm mod}(\lambda_i + 1, 2)$ with probability $\beta$ 
and $\lambda'_i = \lambda_i$ otherwise, for $i=1,\dots, J$. In other words,
$q(\lambda,\lambda') = |\lambda - \lambda'|$ 
is a vector of independent Bernoullis with probability $\beta$, 
and is hence symmetric, with tunable parameter $\beta$ determining the step-size.
We note that a Gibbs sampler can also be used here.
Now we let $\delta_n = 1/n$, and iterate \eqref{eq:robbins_monro} until convergence.

In order to test the accuracy of our methodology, we simulate the data as follows.
Choose 
$a_{ij}\in \bbR$, for $i=1,\dots, d$ and $j=i,\dots,d$, and let $A_{ij} = a_{ij}$ for $j \geq i$ and $A_{ij}=a_{ji}$ otherwise.
Now consider the true target 
\begin{equation}
\pi(\cdot | a) = \frac1{Z(a)} \exp( - x^T A x) \, .
\end{equation}
Similarly, let $\Lambda$ be defined as $\Lambda_{ij} = \lambda_{ij}$ for $j \geq i$ and $\Lambda_{ij}=\lambda_{ji}$ otherwise,
for coefficients $\{\lambda_{ij}\}_{\{i=1,\dots, d, ~j=i,\dots,d\}}$.
We will identify $a, \lambda$ with vectors in $\bbR^{J}$, 
where $J=d(d+1)/2$.
Notice this is of MaxEnt form with $\phi(x) = {\rm vec} [ \{ x_i x_j \}_{\{i=1,\dots, d, ~ j=i,\dots, d\}}]$, 
i.e. first and second moments 
are observed.
We choose $d=10$, let $N=2d$, take $d$ MCMC steps with $\beta=0.6$, and 
define $\delta_n = \delta^1_n 1_{n\leq 2d} + \delta^2_n 1_{n>2d}$, 
where $\delta^1_n = \epsilon n_0/(n_0 + n)$, 
$\delta^2_n = \Gamma^N_J(m-\phi|\lambda^{2d})^{-1} \epsilon n_0/(n_0 + n)$, 
$n_0=5d$, and $\epsilon=1$.
For $n\leq 2d$ the direction of descent is taken as $m - \eta_J^N(\phi)$, 
while for $n>2d$ the full estimator \eqref{eq:unbiased} is used 
$\widehat{F}(\lambda^n) = \Gamma_J^N(m - \phi | \lambda^n )$,  
so that this is an unbiased estimator 
of $Q(m-\phi|\lambda)$. 
This construction prevents large fluctuations in 
$\Gamma^N_J(1|\lambda^{n})$ 
during early outer iterations from precluding 
stabilization of the algorithm.
The truth, reconstruction, and convergence plot are given in Fig. \ref{fig:post}. 

\begin{figure}
\includegraphics[width=0.32\columnwidth]{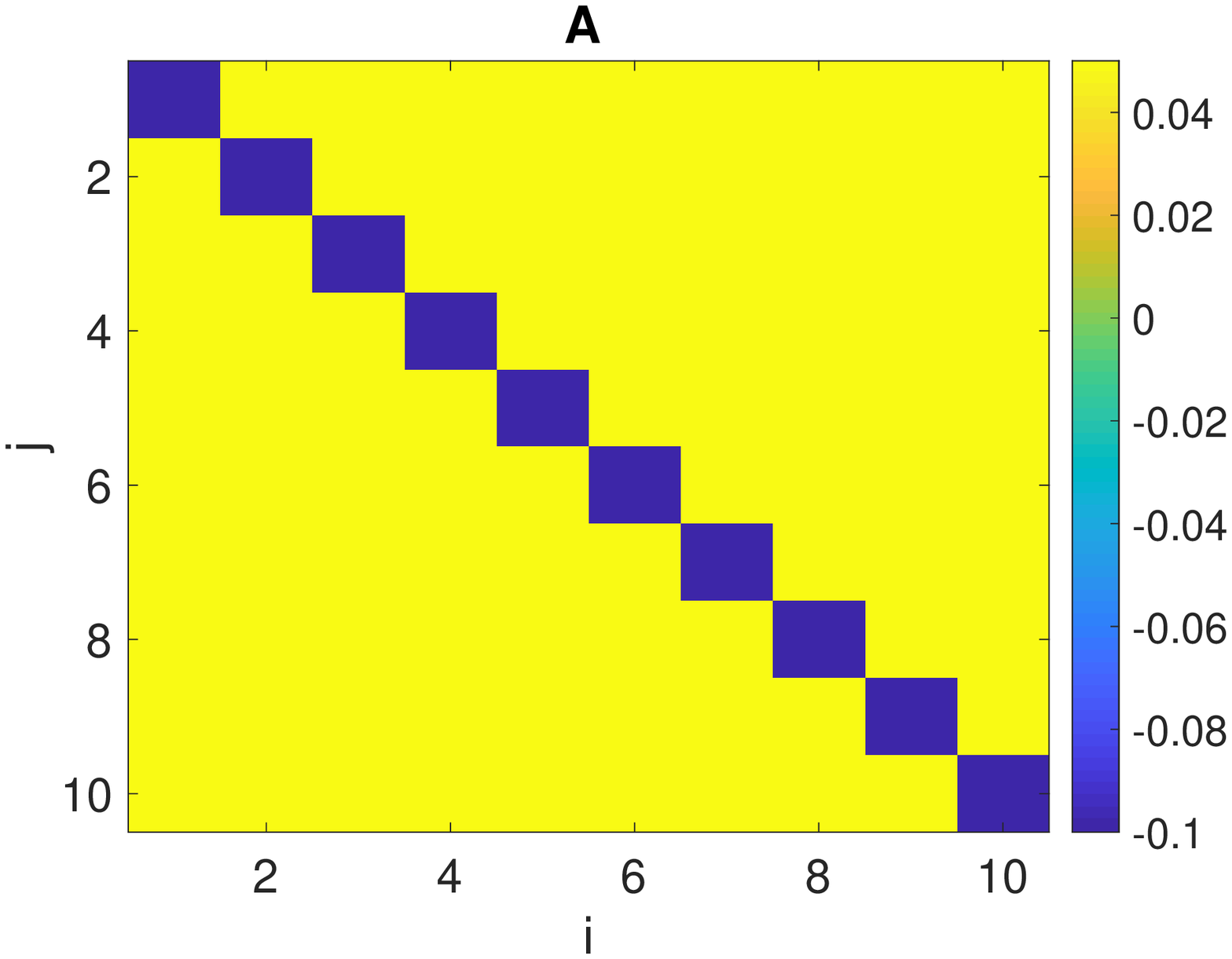}
\includegraphics[width=0.32\columnwidth]{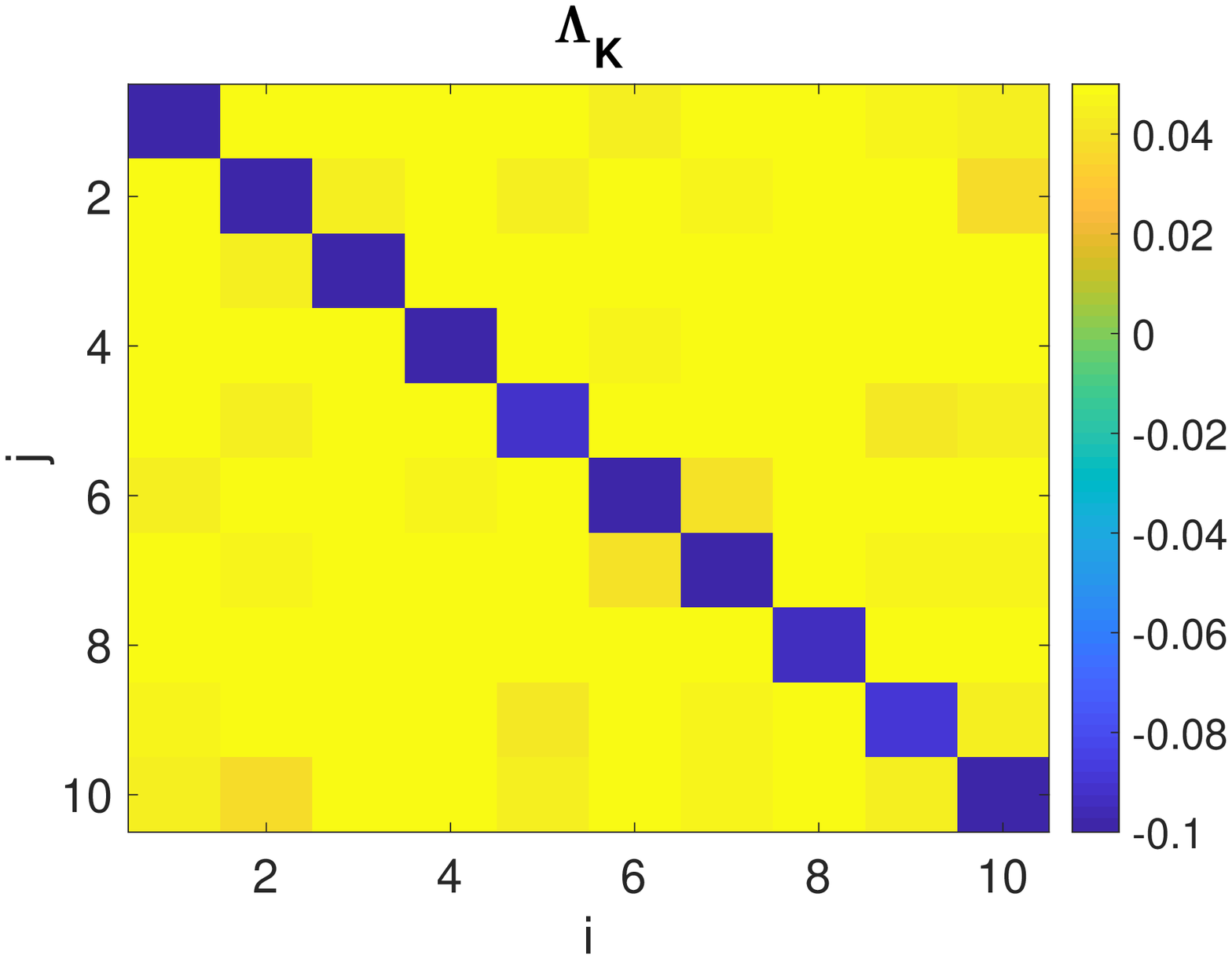}
\includegraphics[width=0.32\columnwidth]{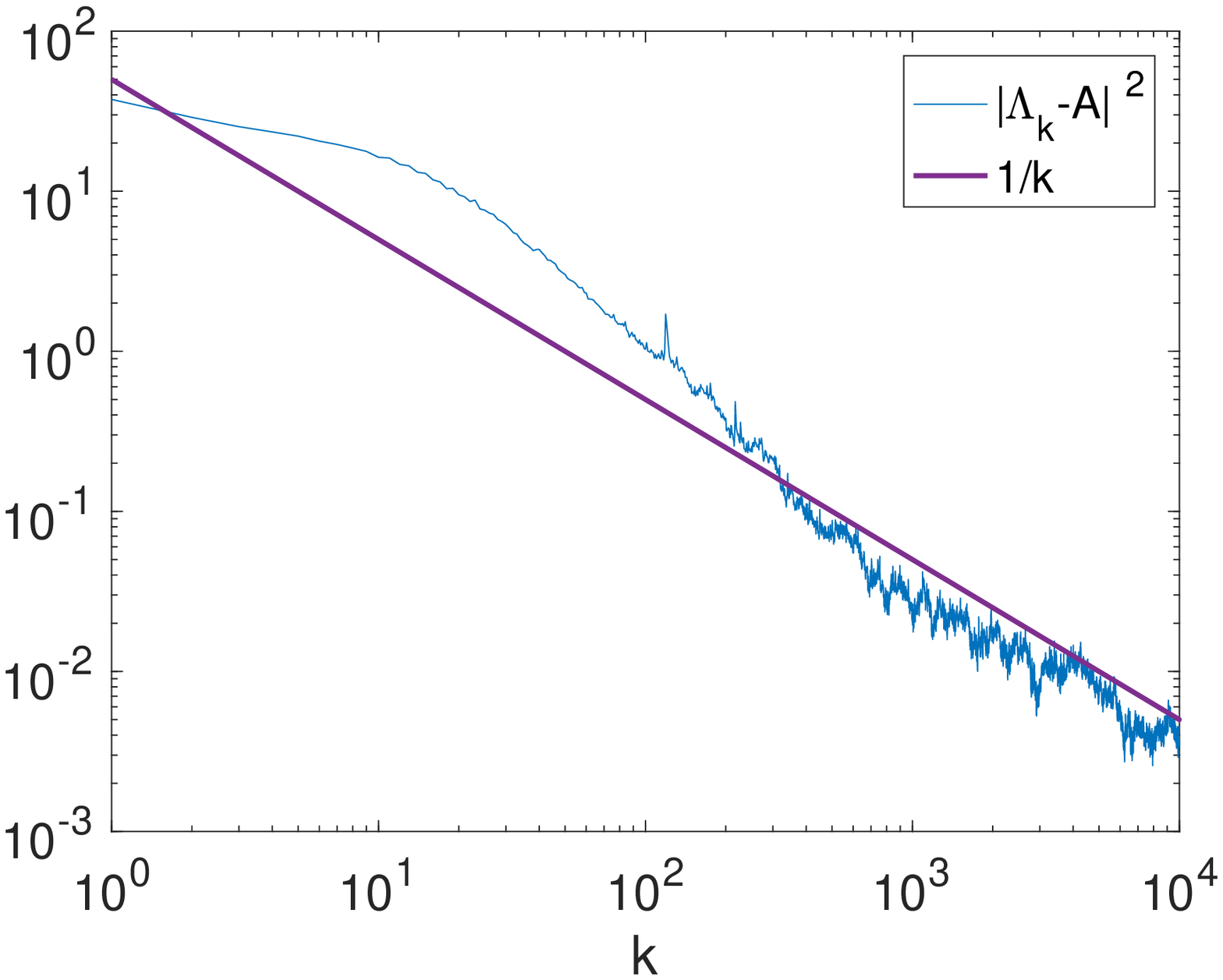}
\caption{Truth (left) and reconstruction after $K=10^4$ steps (middle) for known truth.
The error as a function of iteration is given in the right plot.}
\label{fig:post}\end{figure}

\subsection{Maximum likelihood: observations}
\label{ssec:max_ent_sim_obs}

In this section we explore maximum likelihood estimation with finitely many observations
and empirical moments. We look at the cases of $M=1000$ and $M=50$ observations.
The results are shown in Figure \ref{fig:obsmom}. The top 
two panels correspond to $M=1000$ (left) and $M=50$ (right).
The good news is that we obtain the MLE, up to the bias arising due to the observational error,
quite rapidly: within several hundred iterations for $M=1000$ and several tens for $M=50$.
This is observed in the middle two panels. 
Notice from the plot for $M=50$(right) that the misfit error eventually starts to increase once the algorithm begins to fit noise. 
This is typical and can be expected.
When the algorithm is stopped the error is quite significant.
In the top two panels, we observe that the reconstruction for $M=50$ 
does not look similar to the true $A$ (shown in the left panel of Figure \ref{fig:post}),
however for $M=1000$ it looks quite acceptable.  
Just for a sanity check, observe the bottom two panels which show the 
observed moments for $M=50$ on the left and predicted moments using the reconstruction (top right) on the right.
The agreement is quite decent, and reasonable given the noise level of the observations.
The point is that the error in the observations translates to a much larger error in the parameters $\lambda$.

\begin{figure}
\includegraphics[width=0.4\columnwidth]{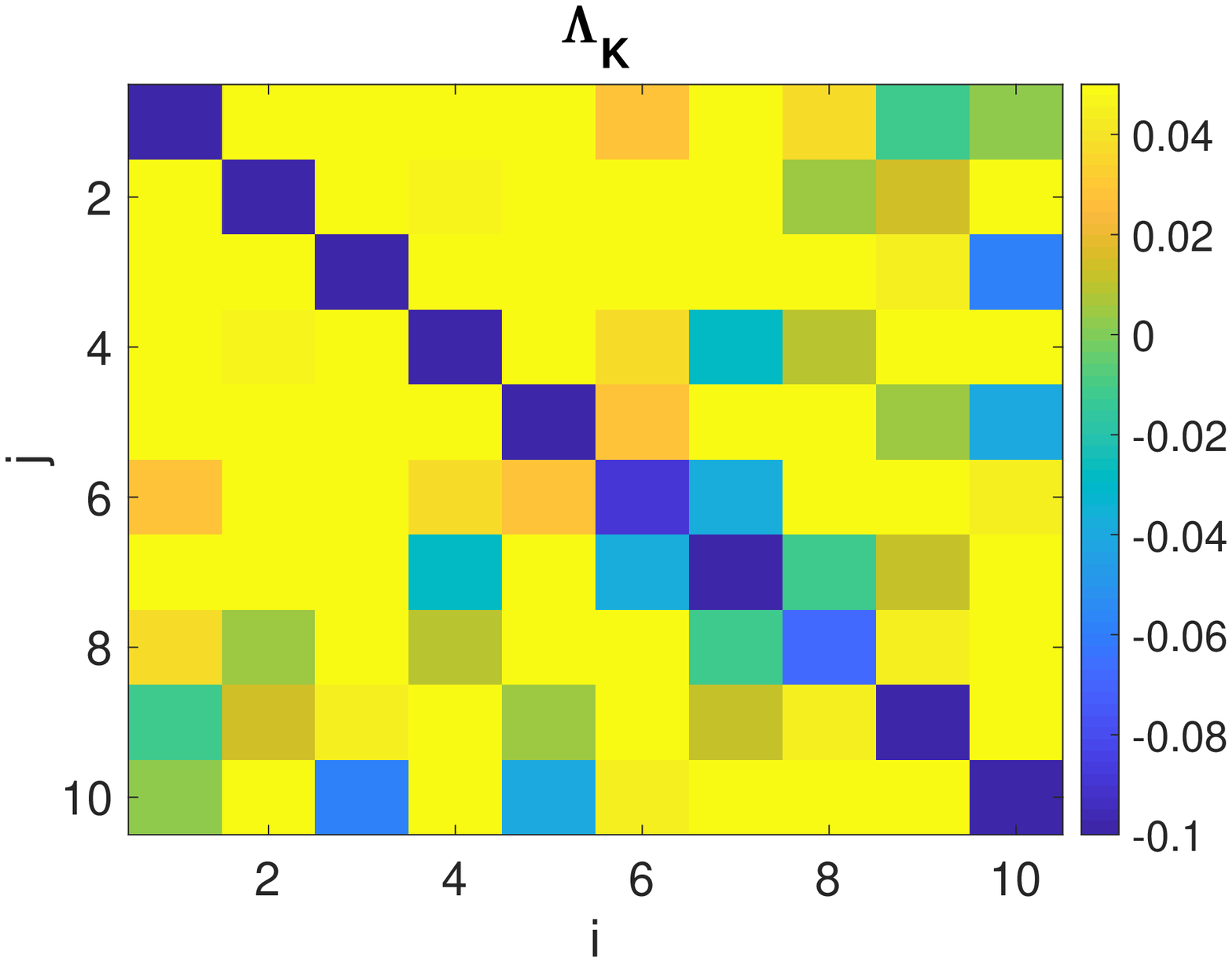}
\includegraphics[width=0.4\columnwidth]{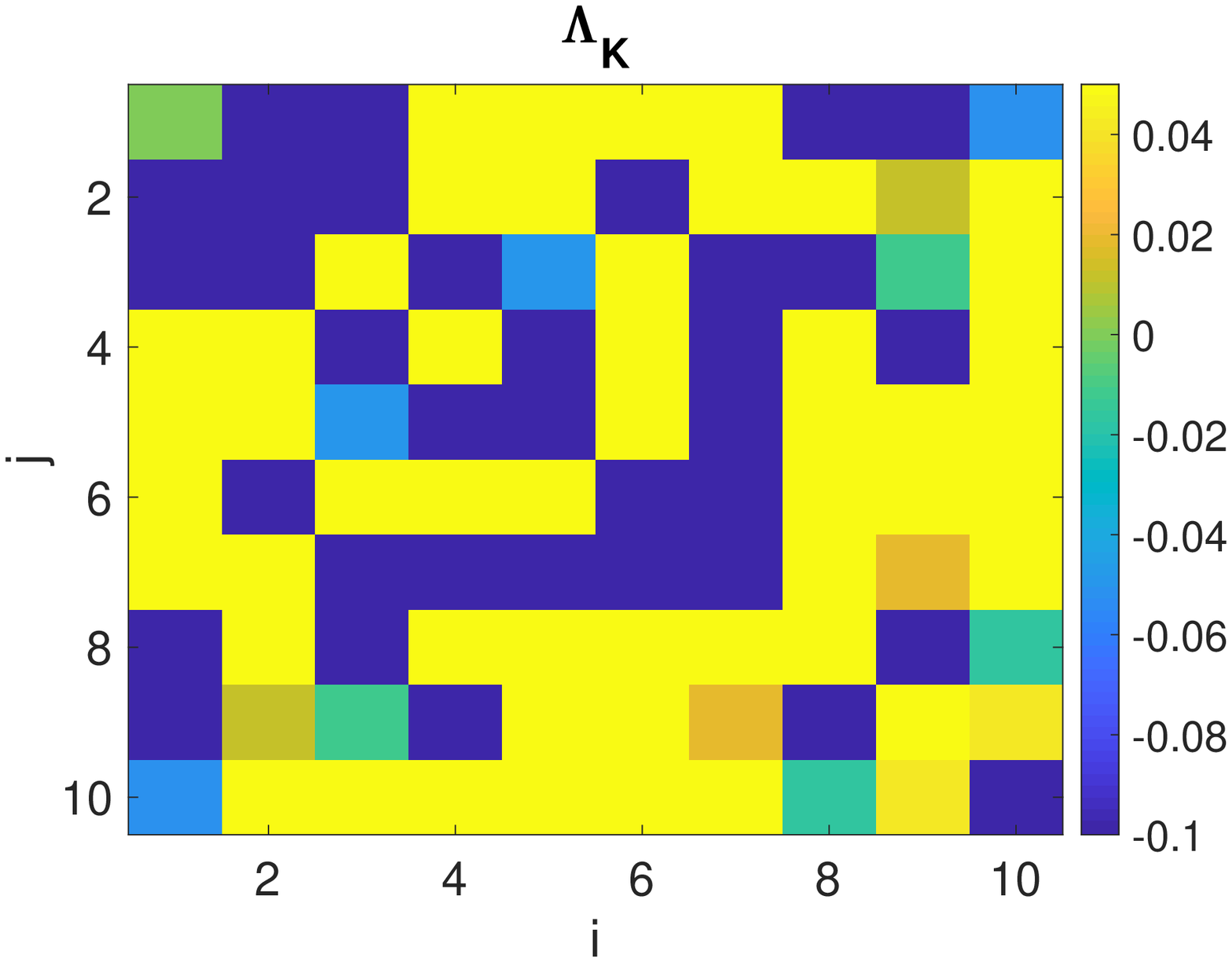} \\
\includegraphics[width=0.4\columnwidth]{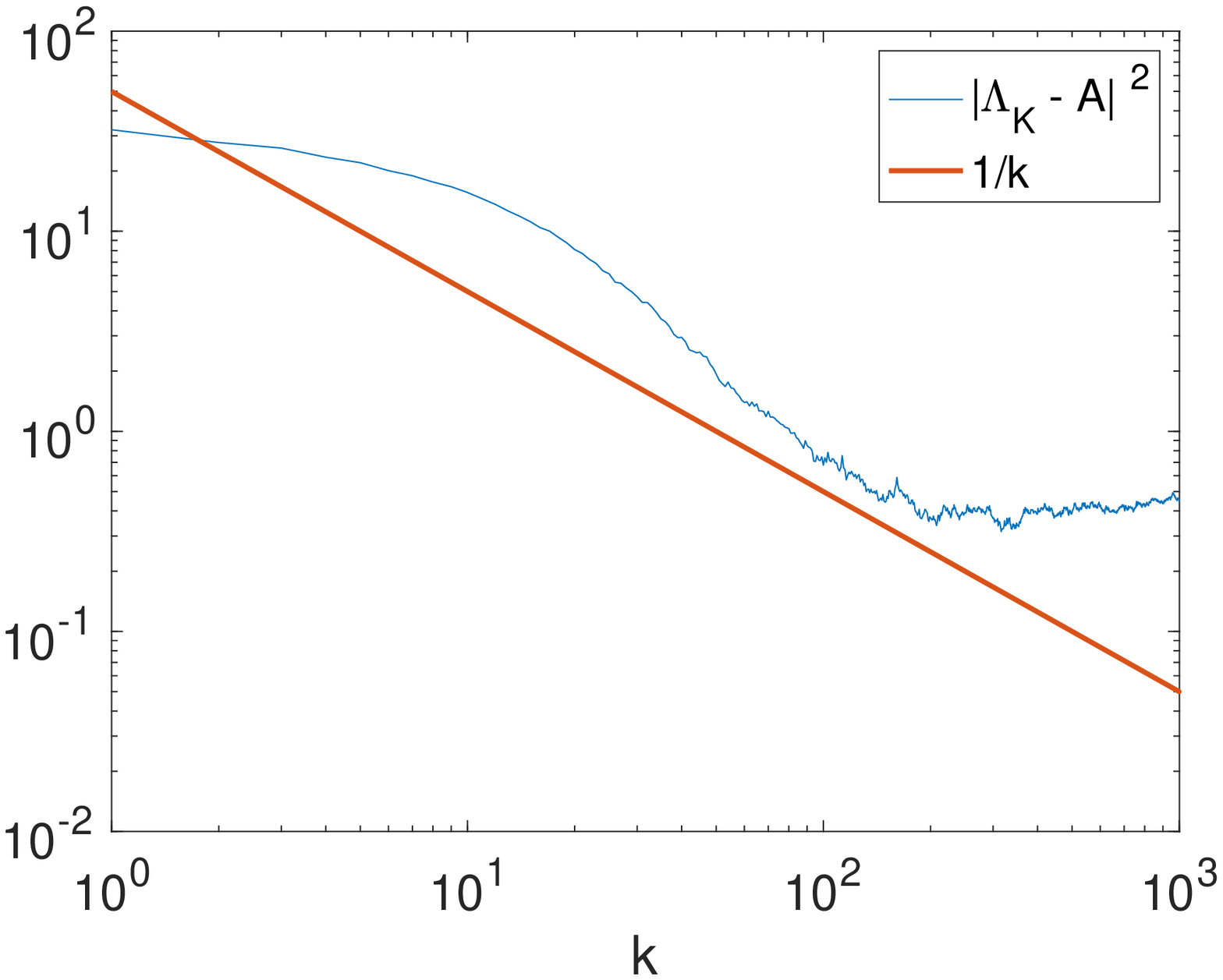}
\includegraphics[width=0.4\columnwidth]{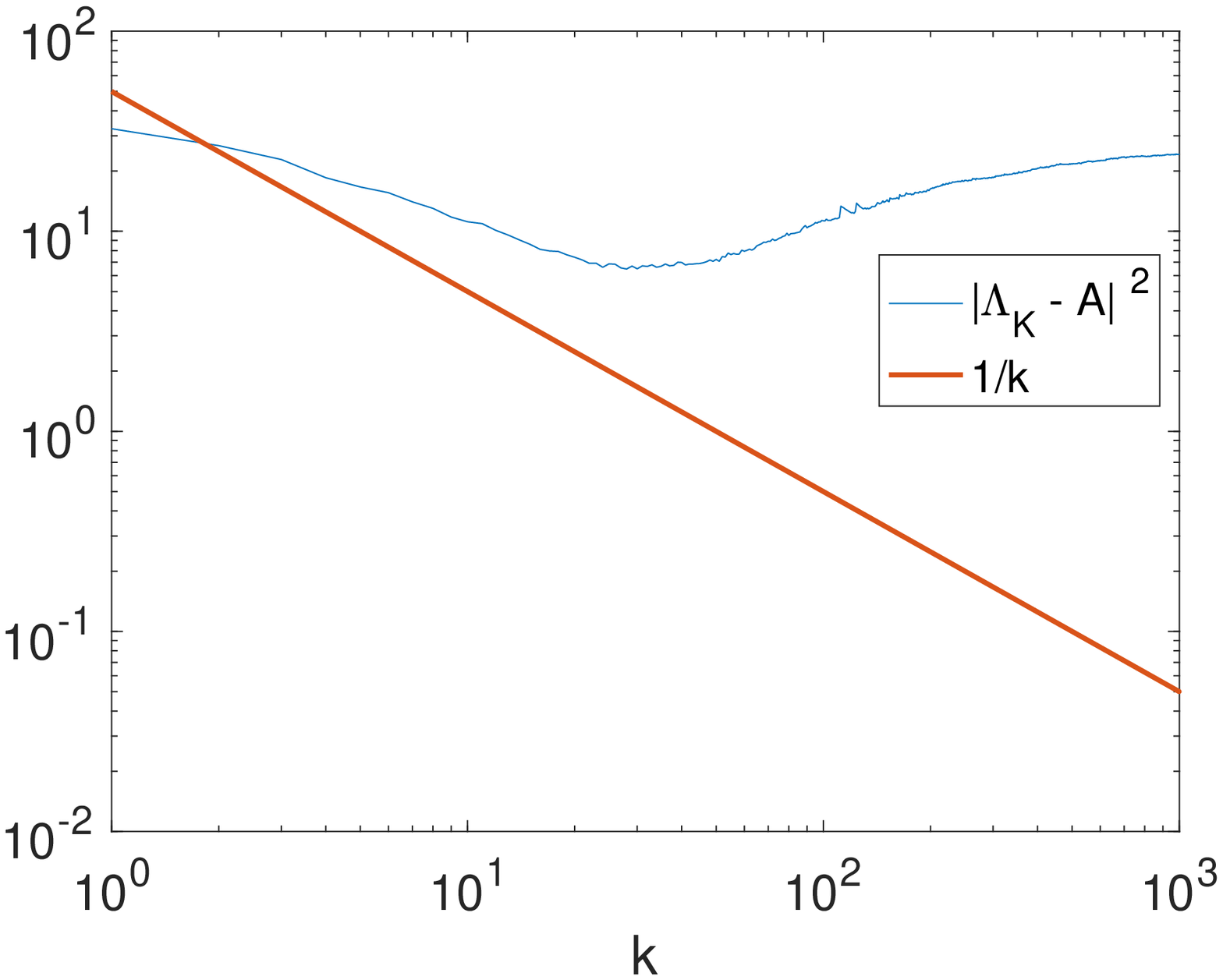}\\
\includegraphics[width=0.4\columnwidth]{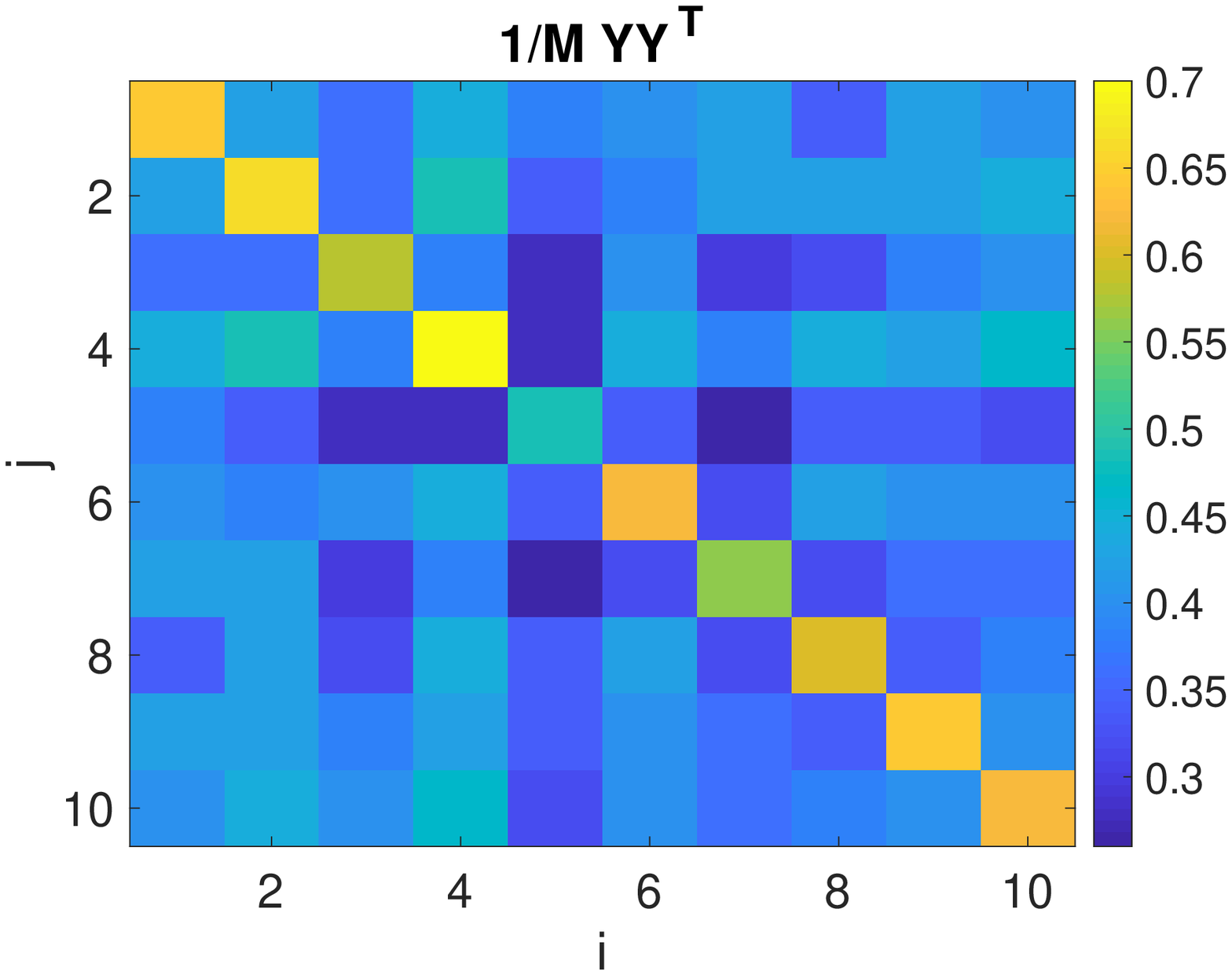}
\includegraphics[width=0.4\columnwidth]{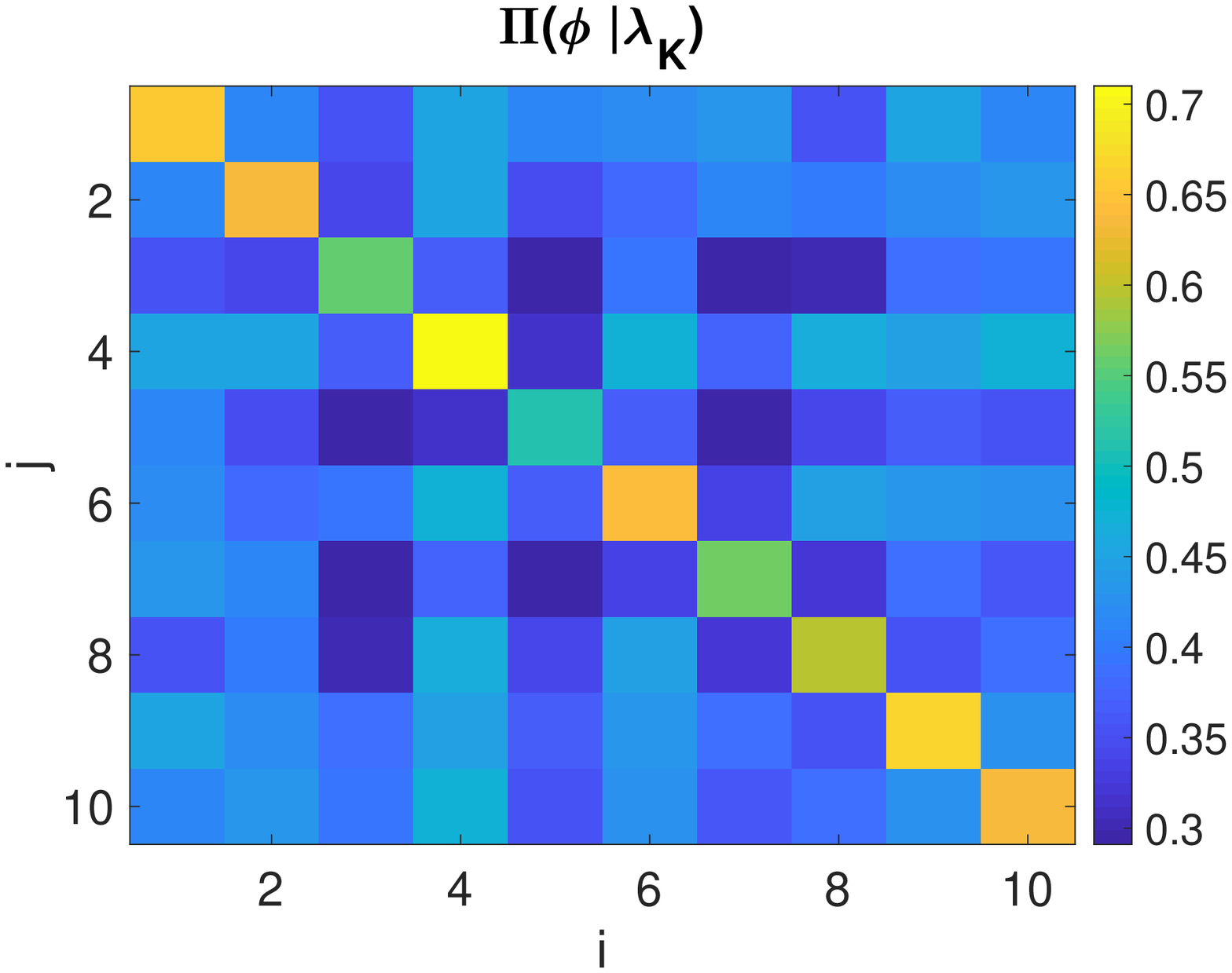}\\
\caption{The reconstruction with $1000$ (left) and $50$ (right) observations are given in the top row,
along with the corresponding convergence plots in the second row. 
The bottom row shows the actual second moments $\widehat m$ (left), with $M=50$ observations, 
which are used to train the model, and the moments under the reconstruction (right). }
\label{fig:obsmom}\end{figure}

\subsection{Uncertainty Quantification using de-biasing}
\label{ssec:uq_debias}

Here we consider sampling from the posterior, with density $\bbP(\lambda | Y)$.
The non-informative improper prior $\bbP(\lambda) \propto 1$ is adopted.

As an initial experiment we ran the SGD algorithm using 
\begin{equation}\label{eq:drift}
\widehat{F}(\lambda) = 
\left (\tilde \eta^{N_L}_J(\widehat{m}-\phi |\lambda) - \tilde \eta^{N_L/4}_J(\widehat{m}-\phi | \lambda)\right)/p_L \, ,
\end{equation}
with $L\sim {\bf p}$ as described in subsection \ref{ssec:debias}, 
the numerator constructed as in \eqref{eq:ais}, and ${\bf p}$ chosen as described in subsection \ref{sssec:smc_debias}, 
with the same 
prototype problem except with $d=4$, and with exact moments.
The results are shown in Figure \ref{fig:sgd} left. 
The middle and right panels
show the results of running the algorithm with the original estimator, and the simple biased but consistent 
estimator $\eta^N_J(m-\phi | \lambda)$, respectively. Note the algorithm with the biased but consistent drift appears to be convergent also. 
As mentioned earlier, this is 
done in practice, but we are not aware of theoretical results verifying this.
Therefore we will proceed with the SGLD using this unbiased estimator, i.e. 
we iterate \eqref{eq:rmsgld} with $P=1$ and $\delta_n=n^{-1/3}$, and drift given by \eqref{eq:drift}, 
with $L \sim \bf{p}$.

\begin{figure}
\includegraphics[width=0.32\columnwidth]{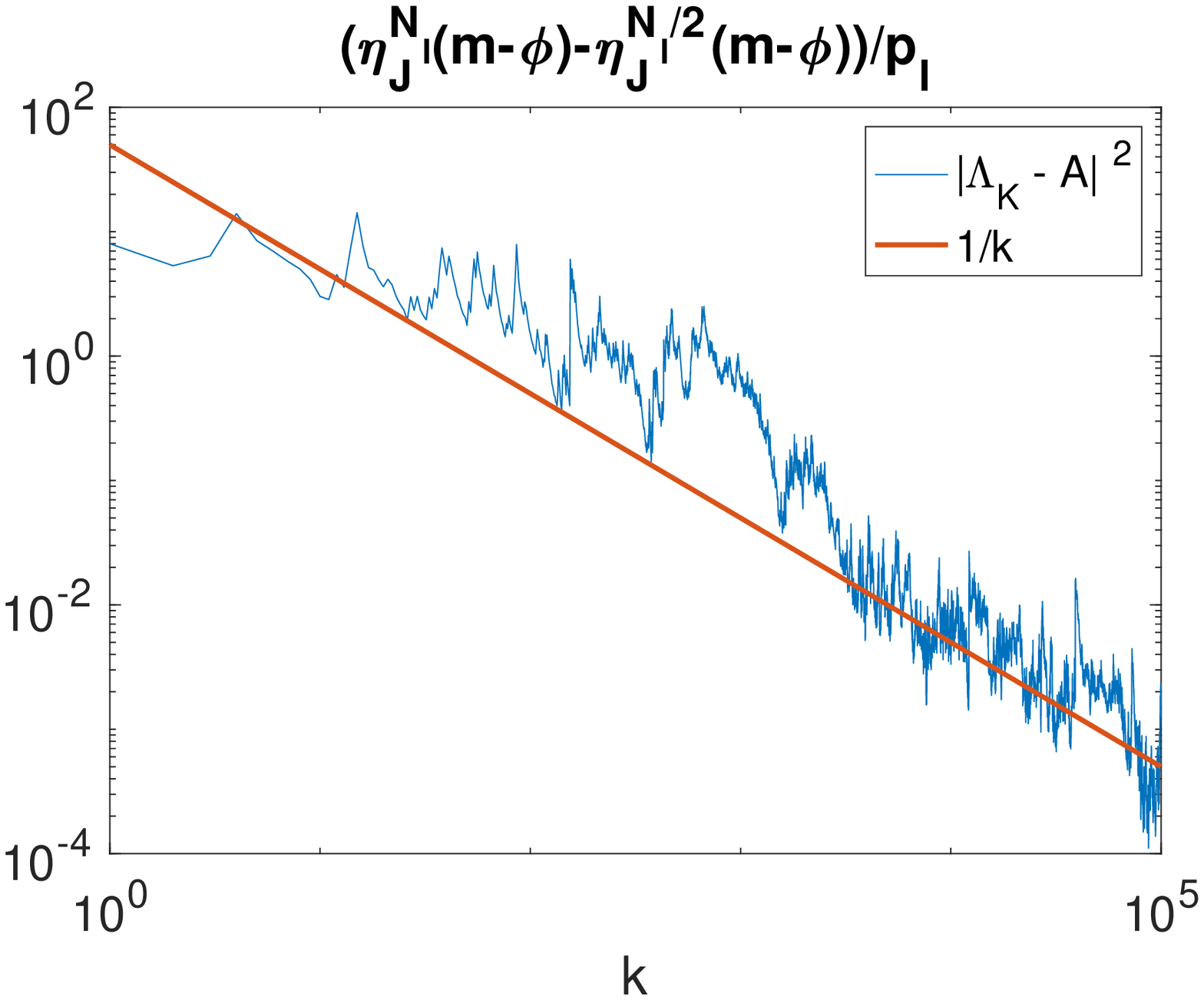}
\includegraphics[width=0.32\columnwidth]{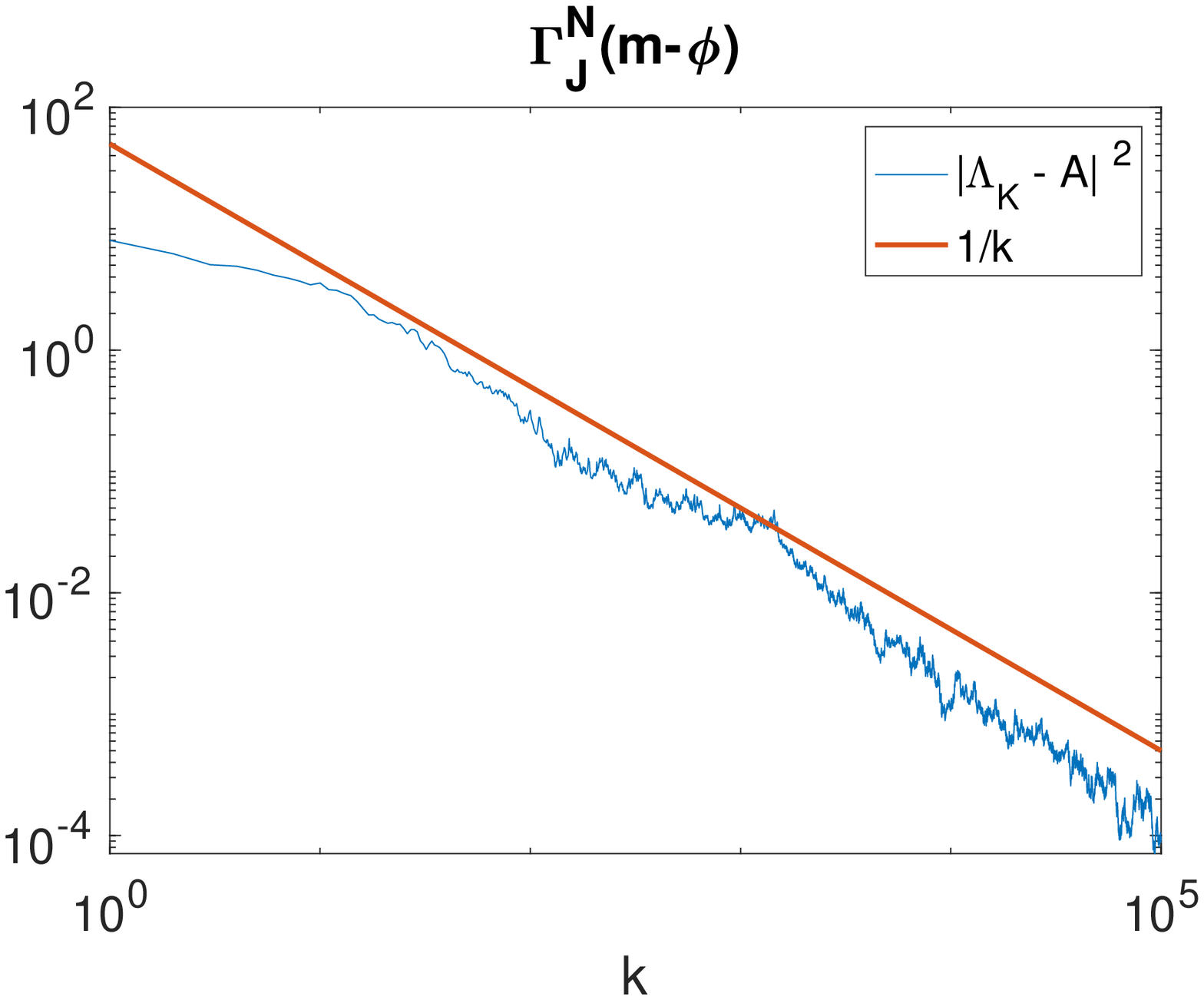} 
\includegraphics[width=0.32\columnwidth]{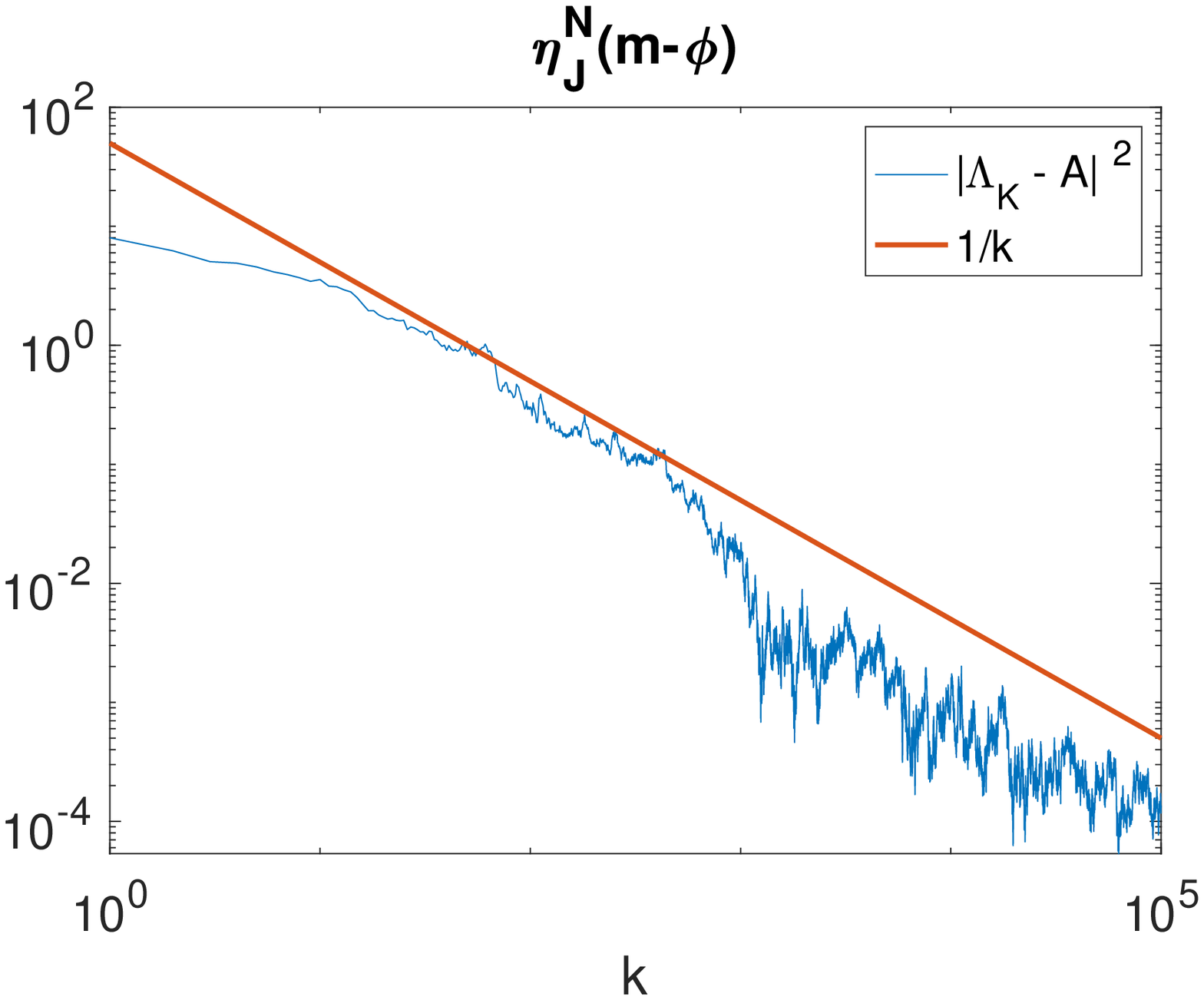}
\caption{Convergence of SGD with the debiased estimator described in Section \ref{ssec:uq_debias} (left), 
the original unbiased estimator $\widehat F$ (middle), 
and the simple consistent but biased estimator (right).}
\label{fig:sgd}\end{figure}

The full SGLD MCMC is run with $M=1000$ observations on the $d=4$ qubit system.
It takes roughly two hours on a laptop to obtain $K=10^6$ samples. 
The pairwise marginals on the diagonal ${\rm diag}(\Lambda)$ are illustrated in Figure \ref{fig:uq}.
The histograms are constructed by discarding the first half of the samples for burn-in, and resampling 
the remaining samples according to their weights $\{\delta_n\}$.
The true value is also indicated in red in the plots.
While the truth is indeed in the region of high probability, and the mean is reasonably close to the truth, 
the spread of the posterior is still quite significant, and more than one would hope for with $M=1000$ observations.
This is consistent with the results of Figure \ref{fig:obsmom}, 
where it is observed that the error in the coefficients is amplified quite a bit in comparison to the error in the 
moments. 
It is however not feasible to compute the pushforward distribution from the coefficients to the moments 
here, since such estimation for a single moment 
would require a separate large sample simulation (e.g. by MCMC or SMC) for each of the $K=10^6$ samples. 
We look at the same simulation with $M=10^6$ observations, and the results are plotted in Figure \ref{fig:uq2}. 
The spread is much tighter, and also there is a strong correlation between the diagonal elements in this case.

\begin{figure}
\includegraphics[width=1\columnwidth]{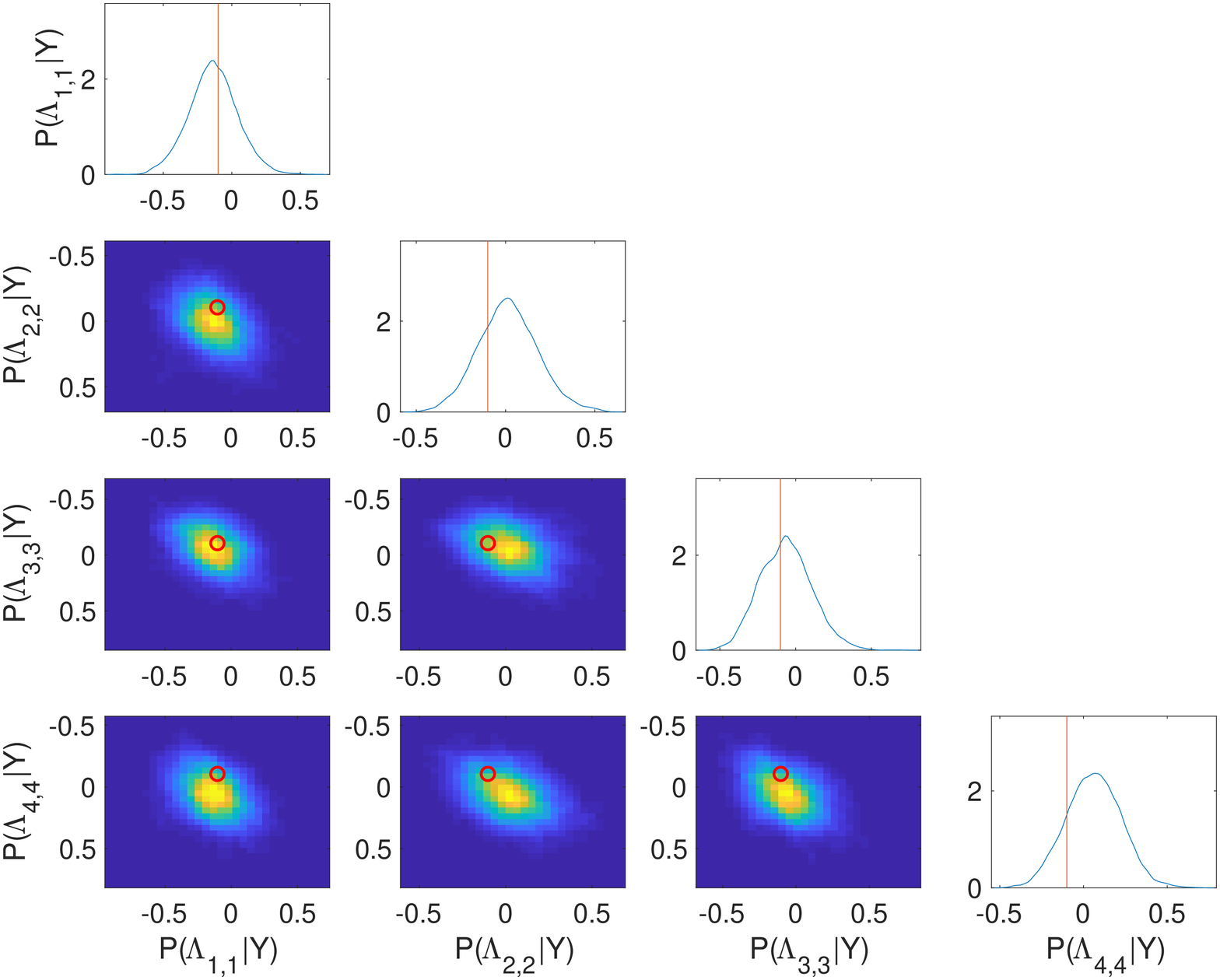}
\caption{Illustration of pairwise marginal UQ for the posterior on the diagonal of $\Lambda$ with $M=1000$ observations and $d=4$ qubits.
The true value of the parameters is indicated in red.}
\label{fig:uq}
\end{figure}

\begin{figure}
\includegraphics[width=1\columnwidth]{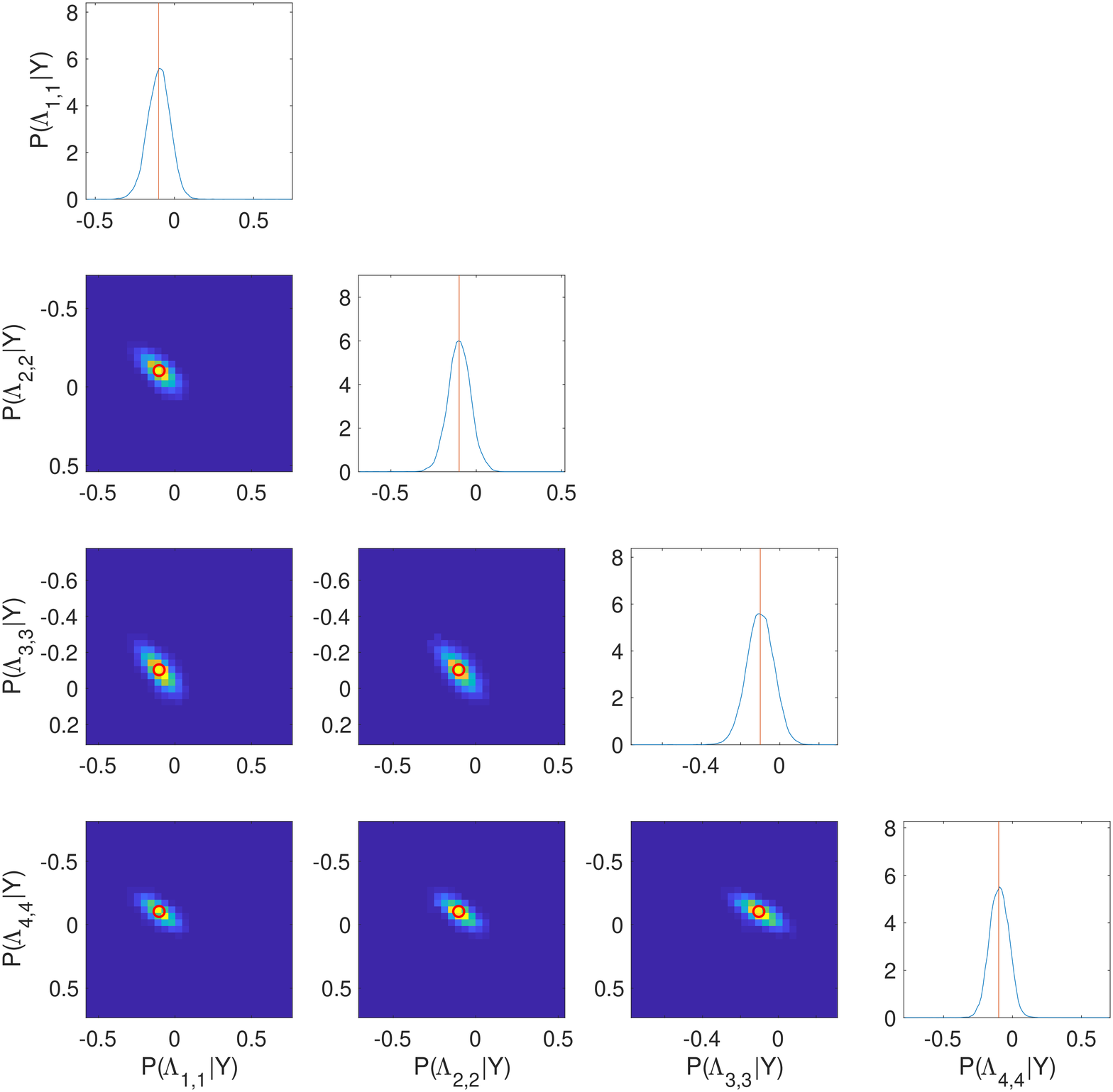}
\caption{Illustration of pairwise marginal UQ for the posterior on the diagonal of $\Lambda$ with $M=10^6$ observations and $d=4$ qubits.
The true value of the parameters is indicated in red.}
\label{fig:uq2}
\end{figure}

%


\section{Summary and Future Directions}

A first attempt has been made to estimate the MaxEnt distribution associated to 
outputs of a quantum computer, and quantify the posterior uncertainty.  This 
task has lead to a novel use of the SMC sampler for construction of an unbiased 
estimator of the intractable drift for use in a Robbins Monro stochastic approximation algorithm,
which can be used to efficiently compute the MaxEnt distribution for exact moments or the MLE distribution
for observed moments.  The Bayesian formulation yields to a doubly-intractable target.
For its solution we use de-biased SMC sampler estimators of the log-likelihood within a SGLD algorithm
to construct consistent posterior expectations. 
All the algorithms are provably convergent, and numerical simulations are provided using (classically) simulated 
data from a toy model. 
This exercise reveals that the posterior uncertainty in the distribution can be significantly amplified with respect to the 
uncertainty in the moments arising from having finitely many observations.
Topics for further research include: 
(0) the extension of the current framework to density matrices,
(i) the incorporation of an explicit error distribution for the observations within the model, 
(ii) exploration of more complex targets, i.e. more qubits, (iii) MLMC acceleration of SGLD, 
(iv) investigation of PDMP MCMC methods as alternatives to SGLD, 
(v) other debiasing strategies and opportunites for acceleration, 
and of course (vi) investigation of outer problems, such as model selection.

\vspace{10pt}
\noindent 
{{\bf Acknowledgements}: This work is supported by the U.S. Department of Energy, Office of Science, Office of Advanced Scientific Computing Research (ASCR) quantum algorithm teams program, under field work proposal number ERKJ333.}

\bibliographystyle{plain} 
\bibliography{references}

\appendix

\section{}
\label{app:1}

The proof of Proposition \ref{prop:unbiasedunno} is given below.

\begin{proof}

Define 
\begin{equation}\label{eq:phi}
\Phi_j(\eta) := \frac{\eta(G_{j-1} M_j)}{\eta(G_{j-1})} \, ,
\end{equation}
and observe that the iterates of the algorithm of section \ref{sec:smc} can be rewritten in the concise form
$$
x^i_j \sim \Phi_j(\eta^N_{j-1}) \,  , \quad {\rm for} ~~ i=1,\dots,N \, ,
$$
where we recall the definition \eqref{eq:empirical} for the empirical measure $\eta^N_{j-1}$.

Notice that 
\begin{equation}\label{eq:unnorm}
\Gamma_J(\varphi)  = 
\int_{E^{d\times (J+1)}} \Gamma_0(dx_0) G_0(x_0) M_1(x_0,dx_1) \dots G_{J-1}(x_{J-1}) M_{J}(x_{J-1}, dx) \varphi(x) \, . 
\end{equation}

By the law of iterated expectations we have
\begin{eqnarray}\label{eq:itex}
\bbE \prod_{j=0}^{J-1} \eta^N_{j}(G_{j}) \eta_J^N(\varphi) & = & 
 \bbE \left[  \bbE \left[ \prod_{j=0}^{J-1} \eta^N_{j}(G_{j}) \eta_J^N(\varphi) \Bigg | x_{0:J-1}^{1:N} \right] \right] \, \\
& = & 
 \bbE \left[  \prod_{j=0}^{J-1} \eta^N_{j}(G_{j}) \bbE \left[  \eta_J^N(\varphi) \Big | x_{0:J-1}^{1:N} \right] \right] \, \\
& = & \bbE \left[ \prod_{j=0}^{J-1} \eta^N_{j}(G_{j}) \Phi_J(\eta_{J-1}^N)(\varphi) \right] \, \\
& = & \bbE \left[ \prod_{j=0}^{J-2} \eta^N_{j}(G_{j}) \eta_{J-1}^N(G_{J-1}M_J \varphi) \right] \, .
\end{eqnarray}
Iterating in this way, it is clear that
$$
\bbE \prod_{j=0}^{J-1} \eta^N_{j}(G_{j}) \eta_J^N(\varphi) = \eta_0(G_0 M_1 \cdots G_{J-1} M_J \varphi) \, .
$$
Noting that $\Gamma_0=\eta_0$, the right-hand side is exactly \eqref{eq:unnorm}.

%

\end{proof}

\end{document}